\documentclass[10pt,a4paper]{article}
\usepackage{Risex-Lab-Paper}
\labsetup{
  name={RISE-X Lab},
  shortname={RISE-X Lab},
  tagline={School of Computer Science,\\Shanghai Jiao Tong University},
  report={},
  date={September 29, 2026},
  website={},
  paper-type={},
  accent={294D73},
  panel={F1F3F6},
  masthead={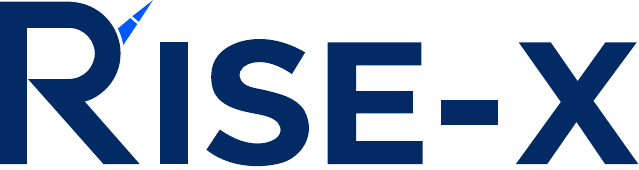},
  masthead-width={25mm}
}

\usepackage{pgfplots}
\pgfplotsset{compat=1.18}

\labsetup{
  text-width={6.2in},       
  title-size={16.5},
  title-leading={27},
  running-title-size={10},
  running-title-leading={0},
  date-size={10},
  date-leading={12}
}
\usepackage{multirow}
\newif\ificlrfinal
\iclrfinaltrue

\providecommand{\masE}{\mathbb{E}}
\providecommand{\masP}{\mathbb{P}}

\providecommand{\masITT}{\operatorname{ITT}}
\providecommand{\masCACE}{\operatorname{CACE}}

\theoremstyle{definition}

\theoremstyle{remark}

\theoremstyle{plain}
\renewcommand{\partmark}{--}

\definecolor{provisionalorange}{RGB}{200,90,0}

\definecolor{draftblue}{RGB}{0,80,190}

\newcolumntype{Y}{>{\raggedright\arraybackslash}X}

\title{%
  \texorpdfstring{%
    \raisebox{-0.10em}
{\includegraphics[height=1.2em]{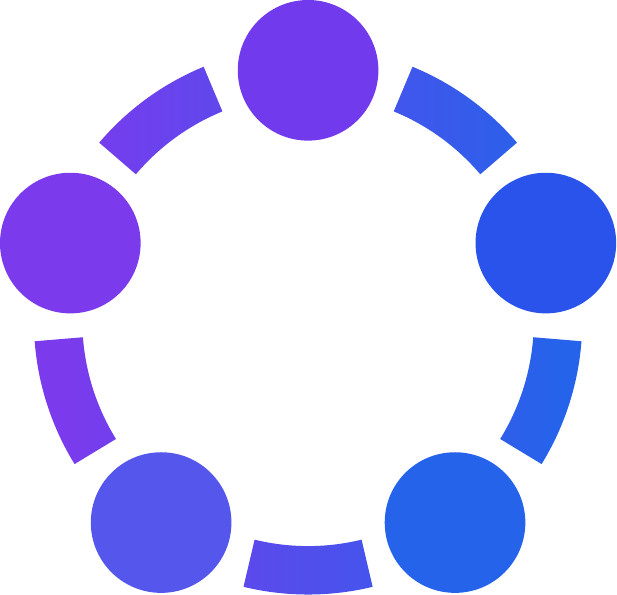}}%
    \hspace{0.20em}%
  }{}%
  OpenCollab: A Multi-Agent Coding Framework
  with Programmable Collaboration and Controllable Runtime
}
\author{\bfseries%
  Chun-Wah Hsu\textsuperscript{1,*}\quad
  Kai Gong\textsuperscript{1,*}\quad
  Yu Wu\textsuperscript{1,*}\quad
  Xianhe Chen\textsuperscript{2}\quad
  Mengyang Liu\textsuperscript{1}\quad
  Jie Li\textsuperscript{3}\\
  \textbf{Hanyu Li\textsuperscript{1} \quad
  Zhixuan Liu\textsuperscript{1}\quad
  Naisheng Tang\textsuperscript{4}\quad
  Jiaying Chi\textsuperscript{1}\quad
  Ziheng Fan\textsuperscript{5}}\quad
  \textbf{Xuning He\textsuperscript{1}
  \\
  Xiaokang Yang\textsuperscript{1}\quad
  Xue Jiang\textsuperscript{6,7}\quad
  Yihong Dong\textsuperscript{1,\Letter}}%
}

\affiliations{%
  \textsuperscript{1}Shanghai Jiao Tong University\quad
  \textsuperscript{2}University of Cambridge\quad
  \textsuperscript{3}Nanyang Technological University\\
  \textsuperscript{4}The University of Hong Kong\quad
  \textsuperscript{5}Imperial College London\quad
  \textsuperscript{6}Peking University\quad
  \textsuperscript{7}Tencent%
}
\paperlinks{\href{https://RISE-X-Lab.github.io/OpenCollab/}{Project page}\enspace /\enspace
 \href{https://github.com/RISE-X-Lab/OpenCollab}{Code}}
\correspondence{Correspondence: \href{mailto:dongyh@sjtu.edu.cn}{dongyh@sjtu.edu.cn}}

\begin{document}
\makelabtitle{%
Multi-agent coding systems are designed to tackle complex software engineering tasks through collaboration. However, existing evaluations typically assume configured organizations are followed faithfully, whereas reality differs. This behavioral gap, combined with differences in underlying system components, prevents clear attribution of observed gains.
To this end, we introduce OpenCollab, a multi-agent coding framework that provides a unified infrastructure for programmable collaboration and controllable runtime. Specifically, OpenCollab unifies organization design, enforces experimental control on a shared runtime, and tracks execution through fine-grained event streams. On this basis, we define Adherence to quantify whether the declared organization is actually realized.
Our experiments reveal that agents collaborate very differently across configurations: changing any single dimension shifts Adherence, from 47.2\% to as high as 97.2\%. Furthermore, extensive agentic coding benchmarks show that OpenCollab (Duo) establishes new SOTA performance compared to the mainstream harnesses such as Mini-SWE-agent, OpenHands, Codex CLI, and Claude Code, showing that a well-designed organization can outperform strong existing harnesses, while OpenCollab's single-agent configuration uses the fewest tokens across all evaluated suites.
OpenCollab establishes a unified multi-agent infrastructure for easy programmable collaboration and controlled causal evaluation.}
\authornote{*}{Equal contribution, work done during research internship at RISE-X Lab, Shanghai Jiao Tong University.}
\authornote{}{\Letter: Corresponding author.}

\section{Introduction}
\label{sec:intro}

Autonomous coding agents have achieved remarkable progress, demonstrating promising capabilities to resolve complex issues in real-world software repositories \citep{jimenez2024swebench,self-planning,wang2024openhands,IO2Code}. As real-world programming tasks grow in complexity, managing vast contexts and multi-stage execution within a single model becomes increasingly challenging. To push the boundaries of performance, the paradigm is shifting toward multi-agent systems \citep{dong2024selfcollaboration,hong2024metagpt,wu2023autogen}. By assigning specialized organization, such as analysts, developers, and testers, these systems tackle intricate tasks through division of labor and interactive coordination. This design assumes that structured organization improves performance, yet configured collaboration may not actually take place.

\begin{figure}[t]
\centering
\includegraphics[width=\linewidth]{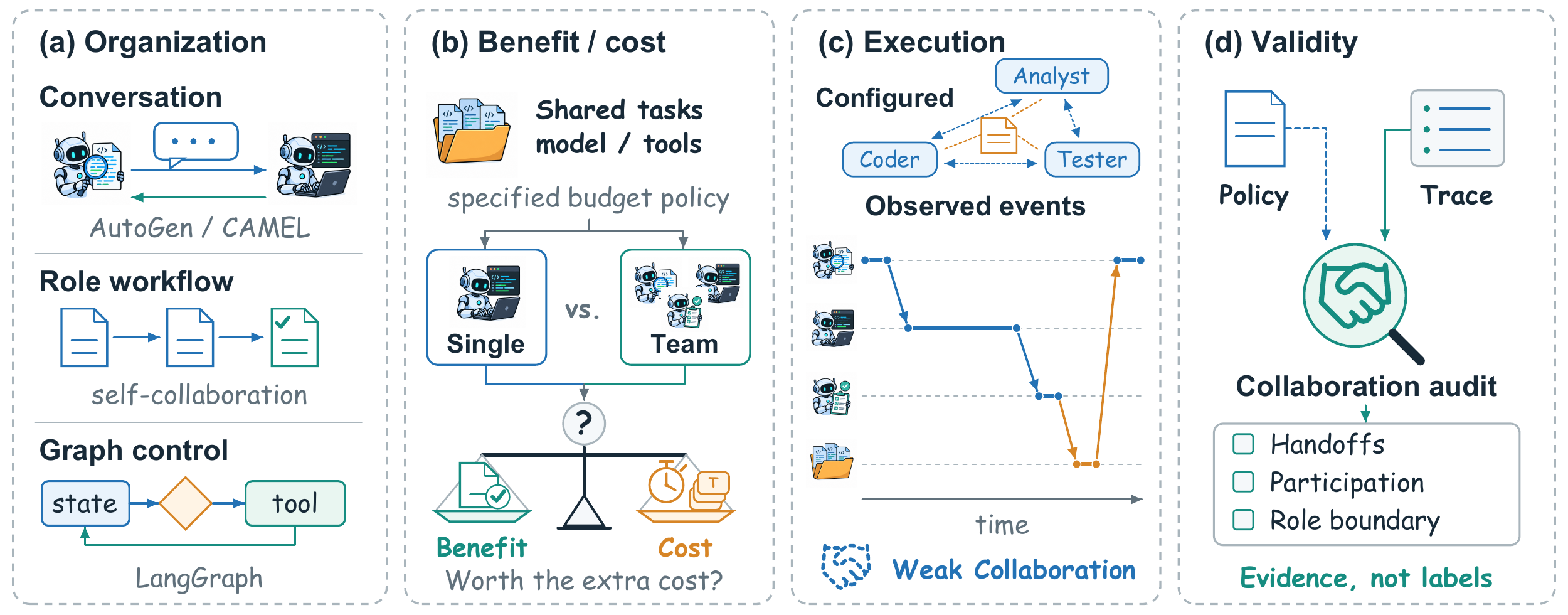}
\caption{Perspectives on multi-agent collaboration:
(a) organization, (b) benefit and cost,
(c) configured versus observed execution, and
(d) collaboration auditing.}
\label{fig:thesis}
\end{figure}

Despite their potential, current multi-agent frameworks lack the infrastructure to evaluate actual effectiveness. First, existing systems lack a unified mechanism to explicitly declare an organizational topology. Second, multi-agent evaluations are rarely controlled, as harness differences and run-to-run variance can obscure what produced a performance gap \citep{zhang2026harness,kaliyev2026noisefloor}. Third, intermediate collaborative processes suffer from poor observability, since benchmarks report only final outcomes \citep{kirgis2026loganalysis} and logs are written for human debugging, not attribution. Finally, this opacity obscures a fundamental mismatch between declared and realized organizations, where agents frequently work in isolation without actual delegation (Figure~\ref{fig:thesis}). Consequently, standard benchmarks cannot determine whether an organization actually helps, because they wouldn't verify whether collaboration took place.

In this paper, we introduce OpenCollab, a multi-agent infrastructure for programmable collaboration and controllable runtime. To enable explicit declaration, OpenCollab provides unified abstractions to statically specify role boundaries and communication topologies. To controlled comparisons, its runtime executes all configurations on a shared substrate, holding non-organizational factors constant. To structured logs, it captures step-by-step actions and resource usage through fine-grained event streams. Building on it, OpenCollab introduces Adherence to quantify the realized organization and, by recording it per run, checks when the complier average causal effect (CACE) \citep{bloom1984noshows,angrist1996iv} can be attributed to the organization.

To systematically evaluate multi-agent systems, we leverage OpenCollab's event streams to dissect their execution across five dimensions (model, tool set, budget, context, and topology). We find that agents collaborate very differently across configurations: changing any single dimension shifts Adherence, from 47.2\% to as high as 97.2\%; at 47.2\%, lead agents monopolize budgets or bypass teammates. Furthermore, our audit of ten mainstream agent frameworks reveals that none natively guarantees controlled experimental comparison or verifies behavioral delivery. Against Claude Code, Codex CLI, OpenHands, and Mini-SWE-agent, OpenCollab establishes new SOTA performance on SWE-bench Pro (68.91\%), Terminal-Bench 2.1 (83.15\%), and DeepSWE (65.49\%). Finally, OpenCollab demonstrates that effective multi-agent collaboration provides substantial gains over working alone. Its collaborative design OpenCollab (Duo) outperforms its single-agent counterpart OpenCollab (Base) by +3.4 points on Terminal-Bench 2.1 and +9.7 points on DeepSWE.

We summarize our main contributions as follows:
\begin{itemize}

    \item We demonstrate that agents collaborate very differently across configurations: changing a single dimension moves the collaboration rate from 47.2\% to as high as 97.2\%.
	
	\item We formalize Adherence metric to quantify the alignment between a declared organization and its actual runtime execution, grounded in causal identification theory. Furthermore, our rigorous audit of ten mainstream agent frameworks reveals that none can compute Adherence for every organization it runs or ensure controlled comparisons.
	
	\item We present OpenCollab, an infrastructure for programmable collaboration and controllable runtime, supporting teams and workflows. It isolates experimental variables and accurately reconstructs execution trajectories via fine-grained event streams.	
	
	\item On extensive mainstream coding benchmarks, our sample multi-agent OpenCollab (Duo) sets new SOTA performance over other mainstream harnesses, and surpasses the single-agent OpenCollab (Base) by +3.4 points on Terminal-Bench and +9.7 on DeepSWE.
\end{itemize}

\section{Motivation}
\label{sec:motivation}

\begin{table}[t!]
\caption{Audit of seven controlled evaluation conditions. \checkmark: by the artifact's own settings and records; \partmark: only through the researcher's code, or in part; $\times$: not met.}
\label{tab:instruments}
\centering
\resizebox{\linewidth}{!}{%
\begin{tabular}{@{}lccccccc@{}}
\toprule
& \multicolumn{5}{c}{\textbf{Held: can the factor be set explicitly?}} &
\multicolumn{2}{c}{\textbf{Verified: can the run be checked?}} \\
\cmidrule(lr){2-6}\cmidrule(lr){7-8}
\textbf{Artifact} & Model & Tools & Budget & Context & Topology & Realized & Compared \\
\midrule
\multicolumn{8}{@{}l}{\textit{Multi Agent frameworks}} \\
AutoGen {\scriptsize(\citealp{wu2023autogen})} & \checkmark & \checkmark & \partmark & \checkmark & \checkmark & $\times$ & $\times$ \\
AG2 {\scriptsize(\citealp{ag2ai2024ag2})} & \partmark & \partmark & $\times$ & \partmark & \partmark & \partmark & $\times$ \\
LangGraph {\scriptsize(\citealp{langchainai2024langgraph})} & \partmark & \partmark & $\times$ & \partmark & \partmark & $\times$ & $\times$ \\
\midrule
\multicolumn{8}{@{}l}{\textit{Coding agents}} \\
SWE-agent {\scriptsize(\citealp{yang2024sweagent})} & \checkmark & \checkmark & \partmark & \checkmark & $\times$ & \checkmark & $\times$ \\
OpenHands {\scriptsize(\citealp{wang2024openhands})} & \checkmark & \checkmark & \partmark & \partmark & \partmark & \partmark & $\times$ \\
Claude~Code {\scriptsize(\citealp{anthropic2026claudecode})} & \checkmark & \checkmark & \partmark & $\times$ & \partmark & \partmark & $\times$ \\
Codex~CLI {\scriptsize(\citealp{openai2026codex})} & \checkmark & \partmark & $\times$ & $\times$ & \partmark & \partmark & $\times$ \\
DeepSeek~Harness {\scriptsize(\citealp{deepseek2026harness})} & \checkmark & \checkmark & $\times$ & \partmark & \partmark & \partmark & $\times$ \\
\midrule
\multicolumn{8}{@{}l}{\textit{Evaluation frameworks}} \\
Inspect~AI {\scriptsize(\citealp{inspect2024})} & \checkmark & \checkmark & \partmark & \partmark & \checkmark & \partmark & \partmark \\
HAL {\scriptsize(\citealp{kapoor2026hal})} & \partmark & $\times$ & $\times$ & $\times$ & $\times$ & \partmark & $\times$ \\
\midrule
OpenCollab (Ours) & \checkmark & \checkmark & \checkmark & \checkmark & \checkmark & \checkmark & \checkmark \\
\bottomrule
\end{tabular}%
}
\end{table}

Deploying multi-agent teams on repository-level tasks typically begins by configuring specialized roles, such as an Analyst, Coder, and Tester \citep{dong2024selfcollaboration}. While intended to divide labor, trajectory analysis reveals that delegation rarely occurs when models decide autonomously. Mostly, the lead agent bypasses configured teammates to resolve issues in isolation, while in other cases, it exhausts its token allowance before issuing any handoff. Consequently, the rate at which the declared organization is realized (formalized as \textit{Adherence}) frequently falls below half as the system collapses into a single-agent loop. We term this divergence between configured and executed organizations the \textit{illusion of collaboration}.

Attempting to resolve this illusion via prompt engineering or log filtering introduces deeper dilemmas. Incrementally mandating delegation in prompts causes Adherence to swing wildly, demonstrating that collaboration is highly sensitive to phrasing rather than organizational structure alone. Conversely, filtering logs to collaborating runs introduces severe post-treatment selection bias. Agents solve easy tasks alone and seek help only when stuck, so this subset skews toward harder tasks. That penalizes the multi-agent arm and invalidates causal claims.

The persistence of this illusion and the failure of simple workarounds stem from underlying architectural limitations in current agent frameworks. First, decentralized resource management allows agents to exhaust budgets locally within isolated execution loops. Because token limits are not governed by a shared pre-call gate across the entire team, lead agents frequently consume the run allowance before reaching a state where delegation could occur. Second, current platforms record execution through conversational prose logs designed for human debugging rather than programmatic analysis. These text transcripts lack role-level attribution, making it difficult to verify which agent executed a specific tool or consumed a given portion of the budget. This opacity prevents automated tools from confirming whether declared topology edges were walked, concealing unexecuted organizations beneath overall benchmark scores.

Addressing these limitations requires three core system capabilities. To prevent unobserved budget exhaustion and eliminate confounders, the framework must operate on a shared runtime enforcing pre-call admission across all arms. To overcome opaque prose logs, execution must be serialized into a structured event stream binding every action and token cost to a distinct role. Finally, to resolve selection bias, the system must compute Adherence from these streams and use CACE to support adherence-aware causal attribution of organizational effects. Together, these capabilities define an infrastructure that unifies programmable collaboration with controlled evaluation, providing the foundation for OpenCollab.

\begin{figure}[t]
	\centering
	\includegraphics[width=\linewidth]{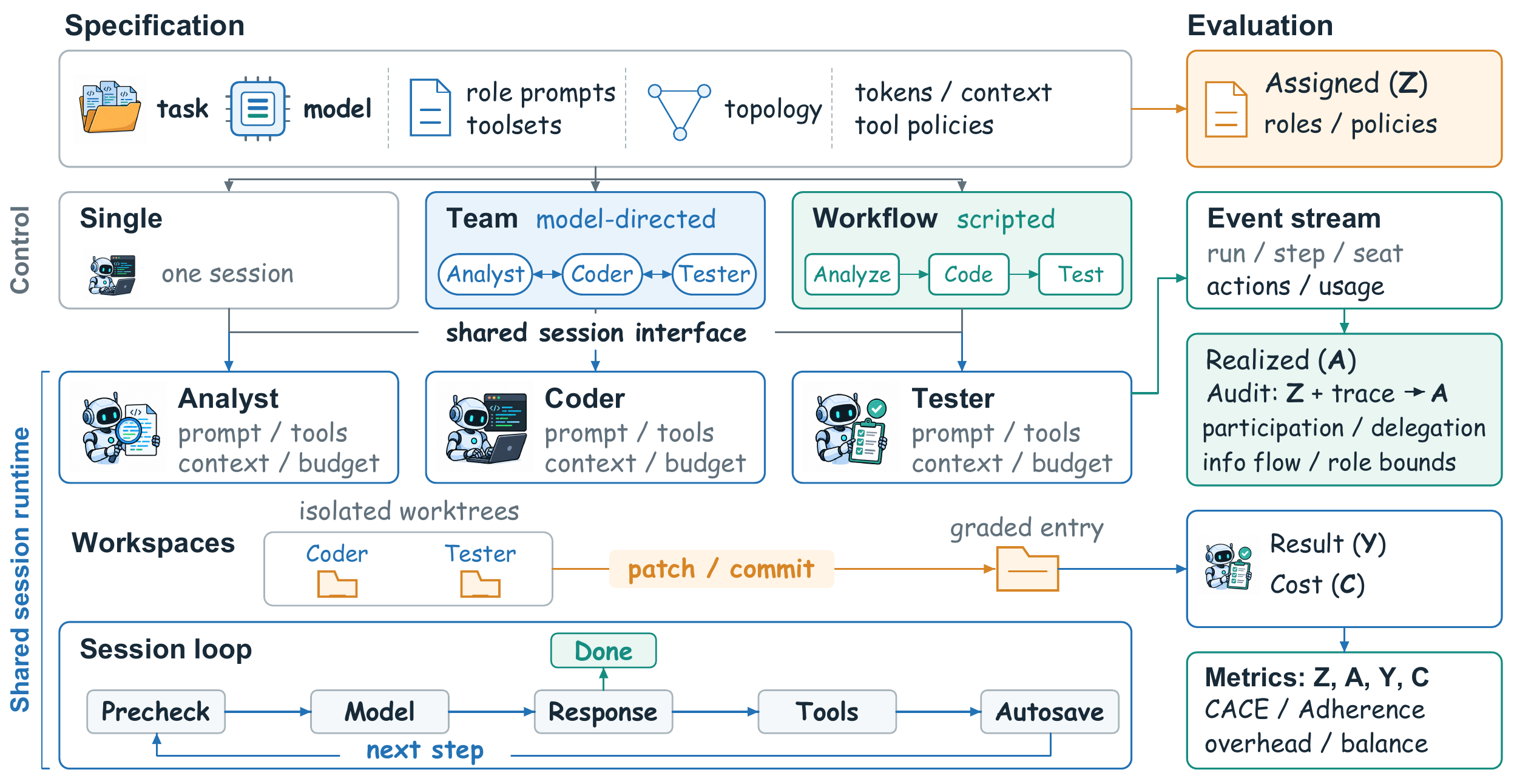}
	\caption{OpenCollab overview. Single, Team, and Workflow controllers share a session runtime. Configuration and execution traces support Adherence and resource costs auditing.}
	\label{fig:mechanism}
\end{figure}

\section{Methodology}

We next describe OpenCollab in detail. In particular, we discuss 1) the architecture of OpenCollab (\S\ref{sec:oc-architecture}), 2) how to program custom collaborative organizations (\S\ref{sec:oc-declare}), 3) how to ensure fair and controlled comparisons (\S\ref{sec:oc-controlled}), 4) how to trace executions using structured event streams (\S\ref{sec:oc-observability}), and 5) how to measure collaboration effectiveness with built-in metrics (\S\ref{sec:oc-metrics}). Figure \ref{fig:mechanism} provides an overview of OpenCollab.

\subsection{Architecture Overview}
\label{sec:oc-architecture}

OpenCollab introduces a decoupled, two-layer architecture, departing from traditional frameworks where each agent encapsulates its own isolated execution loop. As illustrated in Figure~\ref{fig:control-runtime} (Appendix~\ref{app:arch}), this design separates the high-level orchestration of agent interactions from the low-level mechanics of model execution.

At the foundational layer, OpenCollab operates on a shared session runtime. The fundamental unit of execution is a session, representing a single role in an organization. Every session runs through an identical ten-state finite state machine loop detailed in Appendix~\ref{app:arch}. This substrate handles context management, model invocations, and tool execution, ensuring unit consistency across runs.

Built atop the shared runtime, the orchestration layer operates through independent controllers. These controllers govern session topology, lifecycle, and message routing, abstracting control flow away from individual agent implementations. Because all controllers converge on this shared substrate, OpenCollab isolates structural design from low-level runtime mechanics. This decoupling supports diverse organizational paradigms without altering the execution engine, removes infrastructure noise, and facilitates both programmable collaboration and controlled evaluation.

\subsection{Programmable Collaboration}
\label{sec:oc-declare}

To structure multi-agent execution, the orchestration layer supports two primary collaboration paradigms alongside single-agent baselines, summarized as experimental arms in Table~\ref{tab:arms} (Appendix~\ref{app:arch}): a model-directed team mode and a programmatic workflow mode.

The team controller enables a model-directed collaboration paradigm. Users provide a declarative configuration specifying role prompts, accessible toolsets for each role, and a directed communication topology indicating which roles can address one another, following the declaration schema in Appendix~\ref{app:declare}. Once configured, the language model operates within these boundaries, deciding whether, when, and how to delegate subtasks to available teammates. This mode mimics workplace dynamics, allowing multiple agent sessions to cooperate concurrently and pass messages along declared topology edges.

The workflow controller enables programmatic and deterministic agent orchestration. OpenCollab exposes an interface that allows researchers to explicitly govern the collaboration graph. Through simple Python scripts, users can define call orders, fan-out structures, and conditional state transitions across multiple sessions. The script directly controls message routing and termination, ensuring structured execution rather than relying on volatile model decisions.

By separating orchestration from execution, these controllers drastically reduce implementation overhead: reproducing the core multi-agent pattern of \textit{edict} ($\sim$24k lines of code) requires merely 239 lines in OpenCollab. While this enables effortless organizational prototyping, evaluating true collaborative effectiveness still requires strictly controlled conditions.

\subsection{Controlled Evaluation and Comparison}
\label{sec:oc-controlled}

OpenCollab embeds control primitives into the shared runtime to enforce identical conditions across all experimental arms. Specifically, Table~\ref{tab:instruments} defines seven conditions for controlled comparison, comprising five factors held constant at runtime and two structural verification requirements. Table~\ref{tab:configs} (Appendix~\ref{app:five-dims}) details how the five controlled factors guarantee that non-organizational variables remain aligned across arms, including the base model, per-role tool availability, token and cost budgets, context compaction policies, and the communication topology. The two verification requirements ensure auditability by recording each run's declaration next to its realized trajectory and evaluating one against the other. As our audit of ten artifacts in Appendix~\ref{app:matrix} reveals, none of them satisfies all seven conditions, whereas OpenCollab meets all seven.

OpenCollab enforces these conditions through dedicated runtime gates rather than post-hoc filtering. Tool calls and inter-agent messages are intercepted at invocation time, where any request outside the declared role permissions or communication topology is refused and recorded, following the declaration rules in Appendix~\ref{app:declare}. For resource control, Figure~\ref{fig:control-runtime} shows how the runtime enforces an identical per-role token ceiling via pre-call admission within the session state machine. Before each model call, the runtime validates an estimated token reservation against the remaining allowance, halting the session with an explicit termination record if exceeded. Appendix~\ref{app:schema} provides the exact termination codes, which prevent unmetered overruns and cleanly separate runs that exhausted their budget from those that finished voluntarily. These explicit termination records and invocation-time gates form the precise foundation of OpenCollab's execution traceability.

\subsection{Execution Traceability and Observability}
\label{sec:oc-observability}

To capture these records, OpenCollab replaces prose logs with an ordered, append-only event stream. Throughout an execution, every operational step emits a structured record containing a monotonic step index, a timestamp, a run identifier, a role identifier, an event type, and a typed payload. As summarized in Table~\ref{tab:events} (Appendix~\ref{app:schema}), the runtime defines seven distinct event types to capture the full lifecycle of each role, covering model calls, tool executions, refused actions, session terminations, spawn attempts, worktree modifications, and context compaction triggers.

Because every event is explicitly stamped with its role identifier, the event stream allows automated scripts to reconstruct the realized execution topology and attribute resource costs without parsing natural language text. To ensure physical delivery attribution on the filesystem, non-entry roles in a team operate in isolated Git worktrees over a shared object store (Appendix~\ref{app:arch}). Only the entry role workspace is graded upon run completion. Appendix~\ref{app:schema} details this event schema.

\subsection{Evaluation of Cooperation Effectiveness}
\label{sec:oc-metrics}
\label{sec:oc-adherence}

Leveraging the structured event stream, OpenCollab extracts runtime variables to compute automated metrics, addressing the evaluation biases outlined in Section 2. Each run $r$ is represented by the assigned configuration $Z_r \in \{\text{Single}, \text{Workflow}, \text{Team}\}$, the realized-organization indicator $A_r \in \{0, 1\}$, the task result $Y_r \in \{0, 1\}$, and the token cost $C_s(r)$ for each role $s \in \mathcal{S}$, where $C(r) = \sum_{s \in \mathcal{S}} C_s(r)$ is the total run cost. For Team runs, $A_r = \mathrm{adh}(r)$, the full structural-adherence indicator defined below.

\textbf{Structural Adherence.} Adherence establishes internal validity by verifying that an assigned organizational topology was delivered at runtime. Table~\ref{tab:adherence} (Appendix~\ref{app:axes}) defines six verification axes: participation, delegation, role boundary, budget sharing, information flow, and context policy. For each applicable axis $x$, the runtime produces an objective audit verdict $v_x(r) \in \{\text{adherent}, \text{deviant}, \text{unverifiable}\}$. A run is judged adherent if and only if every applicable axis passes:
$
	\mathrm{adh}(r) = \prod_{x \in \mathcal{X}(Z_r)} \mathbf{1}[v_x(r) = \text{adherent}],
$
where $\mathcal{X}(Z_r)$ denotes the set of applicable verification axes for assignment $Z_r$. The Adherence rate observed across all runs $\mathcal{R}_a$ within an arm, the sample analogue of $\alpha_{\mathrm{adh}} = \Pr(A(z) = 1)$ (Appendix~\ref{mas:setup}), is:
\begin{equation}\label{eq:alpha}
	\widehat{\alpha}_{\mathrm{adh}} = \frac{1}{|\mathcal{R}_a|} \sum_{r \in \mathcal{R}_a} \mathrm{adh}(r).
\end{equation}
A deviant or unverifiable verdict on any applicable axis gives $\mathrm{adh}(r) = 0$, so delegation alone does not make a run adherent: a run in which only one of two required teammates acts has $\mathrm{adh}(r) = 0$. Because $A_r = \mathrm{adh}(r)$ on Team runs, $\widehat{\alpha}_{\mathrm{adh}}$ is the empirical rate at which the assigned Team organization is realized.

\textbf{Configuration-level complier effect.} To resolve the dilution of observed task success gains, OpenCollab divides the Pass@1 difference from Single (the intention-to-treat effect, ITT) by the Adherence rate, giving the complier average causal effect (CACE):
\begin{equation}\label{eq:wald-ratio}
    \widehat{\mathrm{CACE}} := \widehat{\mathrm{ITT}} \,/\, \widehat{\alpha}_{\mathrm{adh}} = \Delta\widehat{\mathrm{Pass@1}} \,/\, \widehat{\alpha}_{\mathrm{adh}}.
\end{equation}
Under mean exclusion, $\delta_0 = 0$ (Eq.~\eqref{mas:eq:exclusion}), the population ratio $\mathrm{ITT}/\alpha_{\mathrm{adh}}$ identifies the configuration-level complier effect: the effect of assigning the target Team configuration rather than Single among runs in which the assigned Team organization is fully realized, not an isolated effect of coordination itself.
This causal interpretation relies on two conditions. One-sided compliance holds because the Single arm cannot realize the target Team organization (Appendix~\ref{mas:setup}). The mean exclusion restriction requires that, among tasks whose Team execution has $A_r = 0$, assigning the Team configuration rather than Single leaves the mean outcome unchanged. Without it, ITT$/\alpha_{\mathrm{adh}}$ differs from the complier effect by $(1-\alpha_{\mathrm{adh}})/\alpha_{\mathrm{adh}}$ times that $A_r = 0$ contrast (Theorem~\ref{mas:thm:wald}, Proposition~\ref{mas:prop:sensitivity}). While standard evaluations can only assume this restriction, recording $A_r$ per run on paired tasks allows the $A_r = 0$ Team-versus-Single contrast to be estimated, which provides evidence about mean exclusion rather than proof of it (Appendix~\ref{mas:paired}).

\textbf{Cost metrics.} The same records support two cost metrics, the organizational overhead of runs that do not realize the assigned organization and the workload balance of runs that do; Appendix~\ref{app:cost-metrics} defines both, and our experiments do not report them.

\section{Experiments}
\label{sec:experiment}

\begin{table}[t!]
    \caption{Cross-harness evaluation on SWE-bench Pro, Terminal-Bench 2.1, and DeepSWE.}
    \label{tab:benchmark-context}
    \centering
    \small
    \setlength{\tabcolsep}{0pt}
    \renewcommand{\arraystretch}{1.12}
    \begin{tabular*}{\textwidth}{@{\extracolsep{\fill}} l l c c c c @{}}
        \toprule
        \textbf{Benchmark} & \textbf{Harness} & \textbf{Pass@1 (\%) $\uparrow$} & \textbf{Avg. tokens (M) $\downarrow$} & \textbf{Avg. cost (\$) $\downarrow$} & \textbf{Cache hit (\%) $\uparrow$} \\
        \midrule
        \multirow{6}{*}{SWE-bench Pro}
        & Mini-SWE-Agent & 61.66 & 4.16 & 0.90 & 36.65 \\        & OpenHands & 64.77 & 6.16 & 0.87 & 71.82 \\
        & Codex CLI      & 63.73 & 7.31 & 0.73 & 90.27 \\
        & Claude Code    & 58.03 & 9.38 & 2.54 & 12.50 \\
        & OpenCollab (Base)      & 63.21 & 3.89 & 0.41 & 88.03 \\
        & OpenCollab (Duo)       & \textbf{68.91} & 4.79 & 0.47 & 91.50 \\
        \midrule
        \multirow{6}{*}{Terminal-Bench 2.1}
        & Mini-SWE-Agent & 76.40 & 2.97 & 0.45 & 67.36 \\        & OpenHands & 76.40 & 6.63 & 1.03 & 65.17 \\
        & Codex CLI      & 80.90 & 2.86 & 0.26 & 93.88 \\
        & Claude Code    & 77.53 & 21.95 & 2.73 & 79.32 \\
        & OpenCollab (Base) & 79.78 & 1.54 & 0.19 & 78.38 \\
        & OpenCollab (Duo)  & \textbf{83.15} & 4.71 & 0.77 & 61.77 \\
        \midrule
        \multirow{6}{*}{DeepSWE}
        & Mini-SWE-Agent & 61.95 & 19.36 & 3.27 & 58.81 \\        & OpenHands & 62.83 & 26.85 & 3.64 & 74.17 \\
        & Codex CLI      & 47.79 & 15.01 & 1.30 & 96.47 \\
        & Claude Code    & 56.64 & 77.90 & 7.66 & 91.16 \\
        & OpenCollab (Base) & 55.75 & 11.95 & 1.19 & 90.53 \\
        & OpenCollab (Duo)  & \textbf{65.49} & 61.30 & 8.42 & 73.28 \\
        \bottomrule
    \end{tabular*}
\end{table}

Our evaluations focus on two primary objectives: measuring system-level effectiveness and diagnosing internal collaborative mechanics. We evaluate OpenCollab on SWE-bench Pro \citep{deng2025swebenchpro}, Terminal-Bench 2.1 \citep{merrill2026terminal,tbench2026v21}, and DeepSWE \citep{huang2026deepswe}. For cross-harness evaluations, OpenCollab is benchmarked against Claude Code \citep{anthropic2026claudecode}, Codex CLI \citep{openai2026codex}, OpenHands \citep{wang2024openhands}, and Mini-SWE-agent \citep{yang2024sweagent}. We compare three models: DeepSeek-V4.1-Flash, Qwen3.8-Flash, and GPT-5.6-Luna. Extensive experiments and detailed setup can be found in Appendices~\ref{app:protocol}--\ref{app:limits}.

\subsection{System-Level Performance and Organizational Value}
\label{sec:exp-performance}

We first evaluate system-level performance to establish OpenCollab's competitiveness. Table~\ref{tab:benchmark-context} runs six harnesses with GPT-5.6-Luna on SWE-bench Pro, Terminal-Bench 2.1, and DeepSWE. OpenCollab (Duo) attains the highest Pass@1 on all three (68.91\%, 83.15\%, and 65.49\%). OpenCollab (Base) uses the fewest tokens and the lowest cost on all three benchmarks (Appendix~\ref{app:results}).

OpenCollab (Duo) is a Workflow in which code issues every handoff, so both of its Coders run on every task. Paired by task against OpenCollab (Base), it solves 18.58\% of DeepSWE tasks that the single agent fails and fails 8.85\% that the single agent solves (65.49\% vs.\ 55.75\%); on Terminal-Bench 2.1 the split is 10.11\% to 6.74\% (83.15\% vs.\ 79.78\%), too close to separate the two. The workflow spends more tokens, since it runs two complete solving processes and selects between their candidates (Appendix~\ref{app:exp-duo}).

\subsection{Deconstructing Collaboration: Dimensions of Adherence}
\label{sec:exp-adherence}

Standard multi-agent evaluations assume that configured roles operate exactly as declared. We test this by deploying a model-directed team and auditing its execution through OpenCollab's event streams across five configuration dimensions: model, tool set, budget, context, and topology.

\begin{table}[h]
    \caption{Single-dimension ablations relative to the reference team. \textit{Unverified} runs are those terminated before collaboration; brackets report 95\% Clopper--Pearson confidence intervals.}
    \label{tab:five-factors}
    \centering
    \small
    \setlength{\tabcolsep}{0pt}
    \renewcommand{\arraystretch}{1.12}
    \begin{tabular*}{\textwidth}{@{\extracolsep{\fill}} l l c c c c c @{}}
        \toprule
        \textbf{Dimension} & \textbf{Variant} & \textbf{Adherence (\%) $\uparrow$} & \shortstack{\textbf{95\% CI}} & \textbf{Unverified (\%)} & \textbf{Pass@1 (\%) $\uparrow$} & \shortstack{\textbf{Tokens (M)} $\downarrow$} \\
        \midrule
        \multirow{3}{*}{\centering Model}
        & Qwen3.8-Flash & 47.2 & [30.4, 64.5] & 13.9 & 75.0 & 3.05 \\
        & DeepSeek-V4.1-Flash & 66.7 & [49.0, 81.4] & 22.2 & 72.2 & 3.65 \\
        & GPT-5.6-Luna & 97.2 & [85.5, 99.9] & 2.8 & 61.1 & 5.79 \\
        \midrule
        \multirow{3}{*}{\centering Tools}
        & All tools & 47.2 & [30.4, 64.5] & 13.9 & 75.0 & 3.05 \\
        & No edit tools & 86.1 & [70.5, 95.3] & 0.0 & 66.7 & 3.24 \\
        & Read-only & 94.4 & [81.3, 99.3] & 2.8 & 52.8 & 4.24 \\
        \midrule
        \multirow{3}{*}{\centering Budget}
        & 0.5M / role & 47.2 & [30.4, 64.5] & 2.8 & 44.4 & 0.88 \\
        & 2M / role & 75.0 & [57.8, 87.9] & 8.3 & 69.4 & 4.17 \\
        & 4M / role & 77.8 & [60.8, 89.9] & 2.8 & 75.0 & 5.20 \\
        \midrule
        \multirow{3}{*}{\centering Context}
        & Open card & 47.2 & [30.4, 64.5] & 13.9 & 75.0 & 3.05 \\
        & Optional card & 63.9 & [46.2, 79.2] & 16.7 & 61.1 & 3.58 \\
        & Mandatory card & 91.7 & [77.5, 98.2] & 8.3 & 72.2 & 4.04 \\
        \midrule
        \multirow{3}{*}{\centering Topology}
        & Star & 47.2 & [30.4, 64.5] & 13.9 & 75.0 & 3.05 \\
        & Ring + shortcut & 66.7 & [49.0, 81.4] & 13.9 & 66.7 & 3.92 \\
        & Ring & 91.7 & [77.5, 98.2] & 8.3 & 63.9 & 4.21 \\
        \bottomrule
    \end{tabular*}
\end{table}

To isolate the impact of each dimension, we define a baseline reference team using Qwen3.8-Flash, open prompt cards, unrestricted tools, and a budget of 2M tokens per role. For ring topology variants, we instead deploy an analyst-coder-tester structure under the Mandatory card setting (Appendix~\ref{app:five-dims}). Evaluating single-dimension perturbations against this reference team, Table~\ref{tab:five-factors} shows that unconstrained defaults yield only 47.2\% Adherence. However, enforcing constraints like restricted tool boundaries prevents isolated execution, surging Adherence above 90\% and demonstrating that actual collaboration is systematically controllable.

\begin{table}[h]
    \caption{Deployment contrasts and adherence-adjusted estimates across configurations on SWE-bench Pro. ITT is the Pass@1 difference from Single. Adherence is the share of runs with $A_r = 1$. CACE is ITT divided by Adherence. Brackets show sensitivity to the labeling rule for unverifiable verdicts and are not confidence intervals.}
    \label{tab:itt-pp-cace}
    \centering
    \small
    \setlength{\tabcolsep}{0pt}
    \renewcommand{\arraystretch}{1.12}
    \begin{tabular*}{\textwidth}{@{\extracolsep{\fill}} l c c c c @{}}
        \toprule
        \textbf{Configuration} & \textbf{Adherence (\%) $\uparrow$} & \textbf{Pass@1 (\%) $\uparrow$} & \textbf{ITT (\%) $\uparrow$} & \textbf{CACE (\%) $\uparrow$} \\
        \midrule
        Single                & --                      & 69.4 & --      & --                         \\
        \midrule
        Open card (reference) & 47.2 {[}47.2, 61.1{]}   & 75.0 & $+5.6$  & $+11.8$ {[}$+9.1$, $+11.8${]}  \\
        Mandatory card        & 91.7 {[}91.7, 100.0{]}  & 72.2 & $+2.8$  & $+3.0$ {[}$+2.8$, $+3.0${]}    \\
        Budget stated, 2M     & 75.0 {[}75.0, 83.3{]}   & 69.4 & $0.0$   & $0.0$ {[}$0.0$, $0.0${]}       \\
        Read-only             & 94.4 {[}94.4, 97.2{]}   & 52.8 & $-16.7$ & $-17.6$ {[}$-17.6$, $-17.1${]} \\
        \bottomrule
    \end{tabular*}
\end{table}

\subsection{Causal Validity and Trace Diagnostics}
\label{sec:exp-diagnostics}

To compute deployment contrasts, each configuration runs the tasks once and is paired task by task with the Single agent (Table~\ref{tab:itt-pp-cace}). Adherence sets how far CACE (calculated as ITT divided by Adherence, Eq.~\eqref{eq:wald-ratio}) rescales the observed ITT. By default, a run with an unverifiable verdict is assigned $A_r = 0$ (Appendix~\ref{mas:setup}); brackets in our results show the sensitivity of counting these as adherent instead. Under the unconstrained Open card, 47.2\% Adherence rescales an ITT of $+5.6$ to $+11.8$; this gap follows from low Adherence alone and does not by itself indicate a violation of mean exclusion. Reading $+11.8$ as the configuration-level complier effect (Appendix~\ref{mas:appendix}) further requires the $A_r = 0$ Team-versus-Single contrast to have mean zero (Theorem~\ref{mas:thm:wald}). The ratio departs from the complier effect by $(1-\alpha_{\mathrm{adh}})/\alpha_{\mathrm{adh}}$ times this contrast (Proposition~\ref{mas:prop:sensitivity}), a factor that grows as Adherence falls and is small above 90\%. Adherence levels alone cannot verify mean exclusion; only the $A_r = 0$ contrast bears on it. Because OpenCollab records $A_r$ for every run on tasks also run by Single, this contrast is computed rather than assumed (Appendix~\ref{mas:empirical}). Under the Mandatory card, the paired estimate over adherent runs coincides with CACE in this sample ($+3.0$); under the Open card, the $+5.6$ ITT arises on the tasks with $A_r = 0$ (52.8\% of tasks; $+15.8$), whereas the adherent tasks (47.2\%) show $-5.9$, so $+11.8$ should not be read as the complier effect.

OpenCollab's event streams show whether the assigned organization is realized at runtime (Figure~\ref{fig:case-traces}). In the reference team (Figure~\ref{fig:case-traces}b), the adopter briefs both coders but monopolizes execution until exhausting its token budget, never submitting teammate work. Restricting the adopter to Read-only tools (Figure~\ref{fig:case-traces}a) removes its ability to edit files, compelling delegation to Coders and the successful adoption of Coder A's patch. Further trace analyses are in Appendix~\ref{app:case-traces}.

\begin{figure}[h]
	\centering
	\includegraphics[width=\linewidth]{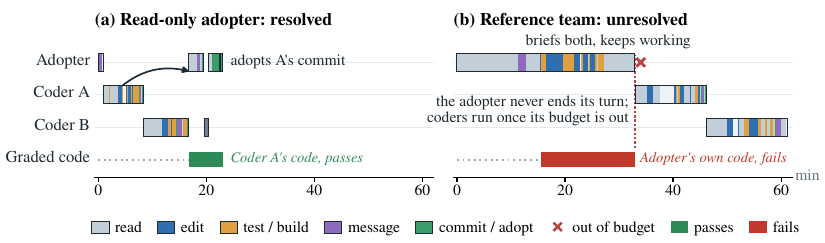}
	\caption{Two execution trajectories from Table~\ref{tab:five-factors} on a shared time axis. Lanes denote agents; blocks represent model calls colored by tool usage.}
	\label{fig:case-traces}
\end{figure}
\section{Related Work}
\label{sec:related}

Autonomous coding agents increasingly adopt multi-agent architectures, drawing on classical organization theory \citep{fox1981organizational,galbraith1974organization,malone1994coordination,decker1993taems,decker1995gpgp,horling2004survey, dong2025survey}. Self-collaboration \citep{dong2024selfcollaboration} is the first multi-agent collaboration coding system by introducing the waterfall model in software engineering methodology. Subsequent work operationalize these concepts through multi-agent interaction \citep{ hong2024metagpt,qian2024chatdev, wu2023autogen}, stateful graphs \citep{langchainai2024langgraph}, or state management \citep{liu2026multi}. These frameworks lower the engineering barrier to deploying complex collaborative topologies. However, existing platforms entangle organizational design with low-level execution, forcing agents to manage clients, contexts, and tools independently. OpenCollab breaks this coupling by isolating orchestration from the runtime, providing a unified substrate to evaluate diverse structures under identical system mechanics.

However, whether multi-agent collaboration outperforms single-agent baselines remains contested. Early work highlighted the benefits of multi-agent debate \citep{du2024debate}, voting ensembles \citep{li2024moreagents}, and repeated sampling \citep{brown2024monkeys}. Conversely, recent studies report mixed or negative returns from adding agents or calls \citep{smit2024mad,wang2024rethinking,chen2024morecalls,zhao2026entropy,bertalanic2026ringelmann,cemri2025whyfail}, even under equal token budgets \citep{tran2026single,kim2026outgrow}. Interacting agent teams often incur substantial coordination overhead \citep{khatua2026cooperbench}, while reported gains often fall below the replication noise floor \citep{kaliyev2026noisefloor}. These conflicting findings are hard to reconcile because comparisons rarely hold budgets, toolsets, or harnesses constant \citep{zhang2026harness}. By enforcing pre-call budget gates and uniform execution primitives, OpenCollab ensures observed performance gaps reflect organizational design rather than infrastructural confounders.

Beyond controlling runtime factors, evaluating collaboration requires confirming that the assigned organization actually materialized during execution. Analyzing interaction networks, \citep{destefanis2026whenagents} show that designated coordinator roles frequently fail to form communication hubs in practice. Related studies examine plan compliance \citep{liu2026planaction}, reversible execution traces \citep{yu2026shepherd}, and role separation \citep{kim2026teambench}, while \citet{kirgis2026loganalysis} emphasize that agent evaluations require log analysis. Yet existing frameworks treat execution as an unverified black box, reporting outcomes without verifying actual communication. OpenCollab bridges this gap via structured event streams and formal Adherence, providing the infrastructure to audit delivery and support causal attribution of organizational gains under explicit identification assumptions.
\section{Conclusion}
\label{sec:conclusion}

Evaluating multi-agent systems is challenging because existing frameworks entangle organizational design with execution mechanics. To address this, we present OpenCollab, a unified multi-agent coding platform that cleanly decouples design from execution. This decoupling allows researchers to easily construct diverse multi-agent systems in a unified manner on a shared substrate with identical runtime semantics. Furthermore, OpenCollab tracks runtime Adherence to verify that the configured collaboration actually occurs. Our findings also reveal that the agents collaborate very differently across configurations, so evaluations that do not verify collaboration cannot attribute their gains to it. Moreover, OpenCollab provides both the unified programmability to effortlessly build agent organizations and the controlled runtime required to attribute their effects causally under the explicit identification. 

\ificlrfinal
\subsubsection*{Acknowledgments}
We sincerely thank Fanjian Su, Lirui Guan, Peize Li, and Weiming Wang for their contributions to the early stages of this work.
\fi



\begingroup
\small
\bibliographystyle{plainnat}
\bibliography{references}
\endgroup

\clearpage
\appendix

\clearpage

\providecommand{\masE}{\mathbb{E}}
\providecommand{\masP}{\mathbb{P}}
\providecommand{\masITT}{\operatorname{ITT}}
\providecommand{\masCACE}{\operatorname{CACE}}
\theoremstyle{plain}
\newtheorem{masthm}{Theorem}[section]
\newtheorem{masprop}[masthm]{Proposition}

\section{Identification theory for the CACE}
\label{mas:appendix}

Terminology: Section~\ref{sec:oc-metrics} and Table~\ref{tab:itt-pp-cace} call the ratio $\masITT/\alpha_{\mathrm{adh}}$ CACE for brevity. The derivations below keep the two apart: the ratio is the Wald ratio (its sample value the Wald estimate), and CACE denotes the complier effect $\masE[\Delta\mid A=1]$ of Appendix~\ref{mas:setup}, which the main text calls the configuration-level complier effect.

Section 3.5 reports the Wald ratio $\masITT/\alpha_{\mathrm{adh}}$, which is read as a complier average causal effect (CACE) only under the mean exclusion restriction below. Here $\alpha_{\mathrm{adh}}$ is the probability that a target Team run satisfies the full operational structural-adherence rule. Every applicable axis must be explicitly judged adherent. The Single arm cannot realize that target organization, giving one-sided compliance. The ratio identifies the effect of the assigned Team configuration relative to Single among fully adherent Team runs when the non-adherent Team-versus-Single contrast has mean zero. Compatible paired records can estimate that contrast; otherwise the exact discrepancy and a sensitivity bound state what the ratio may miss.

\subsection{Setup and notation}
\label{mas:setup}

Using the notation of Section 3.5, write $Z_r\in\{\mathrm{Single},\mathrm{Workflow}, \mathrm{Team}\}$ for the assigned arm, $A_r\in\{0,1\}$ for the realized-organization indicator, and retain outcome $Y_r\in\{0,1\}$, role cost $C_s(r)$, and arm-specific structural adherence $\mathrm{adh}(r)$. Fix a target Team configuration $z$ and write $0$ for the Single baseline. For a task and its pre-specified execution randomness, let $Y(z),Y(0)\in \{0,1\}$ be the configuration-level potential outcomes under the same scoring rule. Their joint law uses a specified coupling of task and run randomness; their marginal means require only comparable arm distributions. We use the same task weights in both arms and assume consistency of observed runs with their assigned potential outcomes, isolated resets, stable model and runtime versions, and complete scoring under a fixed rule. In particular, failures and timeouts are not dropped according to their outcomes.

For a task and its pre-specified execution randomness, let $A(z)\in\{0,1\}$ denote the result of applying the operational structural-adherence rule of Section 3.5 to the execution that would be generated under target Team configuration $z$. With $v_x(z)$ denoting the operational axis verdict generated under that configuration, define
\begin{equation}
 A(z):=\prod_{x\in\mathcal X(z)}
       \mathbf 1\{v_x(z)=\text{adherent}\}.
 \label{mas:eq:delivery}
\end{equation}
For an observed Team run $r$ assigned configuration $z$, consistency gives the observed label $A_r=A(z)=\mathrm{adh}(r)$ for that run's task and execution randomness. An adherent verdict contributes one; a deviant or unverifiable verdict contributes zero. Thus $A_r=1$ exactly when every applicable axis is explicitly judged adherent, and $A_r=0$ otherwise. Every retained Team run has this label, including runs with unverifiable axis verdicts. Write $A:=A(z)$ below and define the population rate $\alpha_{\mathrm{adh}}:=\masP\{A(z)=1\}$. The finite-sample rate $\widehat{\alpha}_{\mathrm{adh}}$ reported in Section 3.5 is the mean of the labels $A_r$ over the evaluated Team runs, the sample analogue of this population rate; likewise $\widehat{\masITT}$ is the difference of observed arm means, and $\widehat{\tau}_{\mathrm{Wald}}=\widehat{\masITT}/ \widehat{\alpha}_{\mathrm{adh}}$ is the sample Wald estimate. Delegation is one component of the structural-adherence criterion. As illustrated in Appendix~\ref{app:axes}, a run can delegate to one teammate while another declared participant never acts; such a run has $A_r=0$ despite delegation. The cost variables do not enter this outcome identification argument.

Define $A(0)=0$ only as a statement that a Single run cannot realize the \emph{target Team} organization. It does not score Single's adherence to its own applicable axes as zero; Single can be adherent to its own topology. The pair $(A(0),A(z))$ is therefore either $(0,0)$ or $(0,1)$. This is the one-sided-compliance, and hence monotonicity, property relevant to the Wald argument \citep{bloom1984noshows,angrist1996iv}.

With $\Delta:=Y(z)-Y(0)$, define
\begin{align}
 \masITT&:=\masE[\Delta],
       \label{mas:eq:itt}\\
 \masCACE&:=\masE[\Delta\mid A=1]
       \quad(\alpha_{\mathrm{adh}}>0),
       \label{mas:eq:cace}\\
 \delta_0&:=\masE[\Delta\mid A=0]
       \quad(\alpha_{\mathrm{adh}}<1).
       \label{mas:eq:delta}
\end{align}
We refer to this configuration-level complier effect as CACE: the effect of assigning the target Team configuration rather than Single among runs that satisfy the full operational structural-adherence rule. It is not, without further intervention assumptions, an isolated effect of coordination itself. The contrast $\delta_0$ concerns the entire non-adherent stratum $A=0$; non-delegating runs generally form only a subset of this stratum.

\subsection{CACE decomposition and Wald identification}
\label{mas:wald-section}

\begin{masthm}[Decomposition and identification of the CACE]
\label{mas:thm:wald}
Under the common consistency and stability conditions above, if $0<\alpha_{\mathrm{adh}}<1$, then
\begin{equation}
 \masITT
 =\alpha_{\mathrm{adh}}\,\masCACE
 +(1-\alpha_{\mathrm{adh}})\,\delta_0.
 \label{mas:eq:decomposition}
\end{equation}
\begin{samepage}
The additional mean exclusion restriction is
\begin{equation}
 \delta_0=\masE[Y(z)-Y(0)\mid A=0]=0
 \label{mas:eq:exclusion}
\end{equation}
When it holds, the ratio identifies CACE:
\begin{equation}
 \masCACE=\frac{\masITT}{\alpha_{\mathrm{adh}}}.
 \label{mas:eq:wald}
\end{equation}
\end{samepage}
If $\alpha_{\mathrm{adh}}=1$, the same equality holds without a restriction on the empty $A=0$ stratum. If $\alpha_{\mathrm{adh}}=0$, the CACE is undefined.
\end{masthm}

\begin{proof}
The law of total expectation, applied to $\Delta$, gives
\begin{align*}
 \masE[\Delta]
 &=\masP(A=1)\masE[\Delta\mid A=1]
   +\masP(A=0)\masE[\Delta\mid A=0]\\
 &=\alpha_{\mathrm{adh}}\,\masCACE
   +(1-\alpha_{\mathrm{adh}})\,\delta_0.
\end{align*}
Substituting Eq.~\eqref{mas:eq:exclusion} and dividing by $\alpha_{\mathrm{adh}}>0$ proves Eq.~\eqref{mas:eq:wald}. For $\alpha_{\mathrm{adh}}=1$, $A=1$ almost surely, so $\masITT=\masCACE$ directly.
\end{proof}

The decomposition follows simply by conditioning on $A(z)$ and does not use $A(0)=0$. The latter condition supports the one-sided-compliance and ``complier'' interpretation; it does not itself imply the Wald ratio. Mean exclusion is the additional conditional mean restriction on \emph{all} target-Team runs with $A=0$, including those with unverifiable axis verdicts; it does not require equal outcomes on every such run. Thus Eq.~\eqref{mas:eq:wald} is the one-sided Wald identification result under mean exclusion and the common conditions, not an unconditional algebraic identity. The arm means identify $\masITT$ under the common conditions, and the operational Team adherence labels identify $\alpha_{\mathrm{adh}}$; the additional restriction supplies the ratio's CACE interpretation.

\subsection{Paired diagnostic for the exclusion restriction}
\label{mas:paired}

The matched-task design can provide a stronger observation regime than the two arm means. If the recorded tuple $(I,A,Y(z),Y(0))$ is a compatible paired observation from the specified joint law, where $I$ is the task, then
\begin{equation}
 \masCACE=\masE[\Delta\mid A=1],\qquad
 \delta_0=\masE[\Delta\mid A=0]
 \quad\text{when }\masP(A=0)>0.
 \label{mas:eq:paired}
\end{equation}
Hence the $A=0$ paired contrast directly estimates the contrast relevant to mean exclusion, and the $A=1$ paired contrast estimates CACE without using exclusion. A finite-sample estimate of the $A=0$ contrast provides evidence about the restriction rather than establishing exact equality. This diagnostic uses the full structural-adherence label and includes $A=0$ runs in which some delegation occurs. Merely sharing a task identifier does not fix the cross-arm dependence of stochastic outcomes: the seed or another explicitly specified coupling protocol (including the product coupling that makes the two arms' execution randomness independent conditional on the task) must make the baseline outcome in each tuple compatible with the same task-and-run coupling used in the definitions above. An unpaired comparison of adherent Team outcomes with the unconditional Single mean need not estimate Eq.~\eqref{mas:eq:cace}.

\subsection{Sensitivity to exclusion violations}
\label{mas:sensitivity-section}

\begin{masprop}[Exact discrepancy and sensitivity bounds]
\label{mas:prop:sensitivity}
For $0<\alpha_{\mathrm{adh}}<1$, without mean exclusion,
\begin{equation}
 \frac{\masITT}{\alpha_{\mathrm{adh}}}-\masCACE
 =\frac{1-\alpha_{\mathrm{adh}}}{\alpha_{\mathrm{adh}}}\,\delta_0.
 \label{mas:eq:discrepancy}
\end{equation}
If $|\delta_0|\leq D_{\max}$ for a nonnegative scalar $D_{\max}$, distinct from the indicator $A(z)$, the absolute discrepancy is at most $(1-\alpha_{\mathrm{adh}})D_{\max}/\alpha_{\mathrm{adh}}$. For $\epsilon>0$, this bound is at most $\epsilon$ if and only if
\begin{equation}
 \alpha_{\mathrm{adh}}
 \geq\frac{D_{\max}}{D_{\max}+\epsilon}.
 \label{mas:eq:threshold}
\end{equation}
If instead $\delta_0\in[\ell_0,u_0]$, a valid sensitivity interval is
\begin{equation}
 \masCACE\in
 \left[
  \frac{\masITT-(1-\alpha_{\mathrm{adh}})u_0}
       {\alpha_{\mathrm{adh}}},
  \frac{\masITT-(1-\alpha_{\mathrm{adh}})\ell_0}
       {\alpha_{\mathrm{adh}}}
 \right]\cap[-1,1].
 \label{mas:eq:sensitivity}
\end{equation}
\end{masprop}

\begin{proof}
Rearrange Eq.~\eqref{mas:eq:decomposition} to obtain Eq.~\eqref{mas:eq:discrepancy}. Taking absolute values gives the bound; solving $(1-\alpha_{\mathrm{adh}})D_{\max}/\alpha_{\mathrm{adh}} \leq\epsilon$ gives Eq.~\eqref{mas:eq:threshold}. Finally, $\masCACE=\{\masITT-(1-\alpha_{\mathrm{adh}})\delta_0\}/ \alpha_{\mathrm{adh}}$ decreases with $\delta_0$, giving the interval. Binary outcomes constrain the CACE to $[-1,1]$.
\end{proof}

The discrepancy in Eq.~\eqref{mas:eq:discrepancy} is a population identification discrepancy, not the finite-sample bias of a random ratio. The interval is valid but need not be sharp after all observed outcome constraints are imposed. A non-delegating-only diagnostic does not bound $\delta_0$ unless it also covers the remaining $A=0$ runs.

\subsection{Reporting requirements}
\label{mas:reporting}

Table~\ref{mas:tab:requirements} separates the deployment effect from the two ways of learning about configuration-level CACE. All rows require the common scoring, consistency, reset, stability, and task-distribution conditions stated above.

\begin{center}
\refstepcounter{table}\label{mas:tab:requirements}
\small
Table~\thetable: Observation regimes and requirements for the Section 3.5 effect.
\par\smallskip
\begin{tabularx}{\linewidth}{@{}>{\raggedright\arraybackslash}p{0.22\linewidth}>{\raggedright\arraybackslash}p{0.27\linewidth}>{\raggedright\arraybackslash}X@{}}
\toprule
Target & Observable input & Additional condition or diagnostic \\
\midrule
$\masITT=\masE[Y(z)-Y(0)]$
 & Task-weighted arm outcomes
 & Stable, complete execution identifies the deployment contrast. \\
$\masCACE=\masITT/\alpha_{\mathrm{adh}}$
 & Arm outcomes and target-Team labels $A_r$;
   $\alpha_{\mathrm{adh}}=\masP(A=1)$
 & $\alpha_{\mathrm{adh}}>0$, target-Team one-sided compliance $A(0)=0$,
   and mean exclusion $\delta_0=0$. \\
Paired CACE and diagnostic
 & Compatible $(I,A,Y(z),Y(0))$ tuples
 & Directly estimate $\masE[\Delta\mid A=1]$ and, when $A=0$ occurs,
   the exclusion-relevant contrast $\delta_0=\masE[\Delta\mid A=0]$. \\
\bottomrule
\end{tabularx}
\end{center}

A report should retain the task, assigned arm and configuration, outcome, per-axis adherence verdicts (including unverifiable), and the seed or coupling protocol used for any pairing. It should state whether the CACE was obtained from the Wald ratio under mean exclusion or directly from compatible pairs. These are statements about the assigned configuration's effect among Team runs that satisfy the full operational structural-adherence rule, not a separate coordination-only estimand.

\subsection{Empirical check on the adherence experiment}
\label{mas:empirical}

Table~\ref{mas:tab:empirical} computes the two paired contrasts of Appendix~\ref{mas:paired} for the four configurations of Table~\ref{tab:itt-pp-cace}. Each configuration is paired with the Single agent on the same SWE-bench Pro tasks, and $A_r$ is the full structural-adherence label of that table, with failed and unverifiable runs at $A_r = 0$. For this analysis we take the cross-arm coupling of Appendix~\ref{mas:setup} to be the product coupling conditional on the task: conditional on the task $I$, the Team and Single execution randomness are independent draws. The task-matched tuples in Table~\ref{mas:tab:empirical} are interpreted under this specified coupling; they do not rely on a shared execution seed. Under it, a Team run's $A_r$ is independent of the Single run's outcome on the same task.

\begin{table}[h]
    \caption{Wald ratio and paired contrasts against the Single agent on SWE-bench Pro. Paired, $A_r=1$ is the mean Team-minus-Single difference over the tasks whose Team run is adherent, the direct estimate of Eq.~\eqref{mas:eq:paired} under the product coupling above; Paired, $A_r=0$ is the same difference over the remaining tasks, the estimate of $\delta_0$. By construction $\widehat{\mathrm{ITT}} = \widehat{\alpha}_{\mathrm{adh}}\cdot(\text{Paired}, A_r{=}1) + (1-\widehat{\alpha}_{\mathrm{adh}})\cdot(\text{Paired}, A_r{=}0)$. Adherence is the share of tasks with $A_r=1$. Differences in percentage points.}
    \label{mas:tab:empirical}
    \centering
    \small
    \begin{tabular}{l c c c c c}
        \toprule
        Configuration & Adherence (\%) & ITT & Wald & Paired, $A_r{=}1$ & Paired, $A_r{=}0$ \\
        \midrule
        Open card         & 47.2 & $+5.6$  & $+11.8$ & $-5.9$  & $+15.8$ \\
        Budget stated, 2M & 75.0 & $0.0$   & $0.0$   & $-3.7$  & $+11.1$ \\
        Mandatory card    & 91.7 & $+2.8$  & $+3.0$  & $+3.0$  & $0.0$   \\
        Read-only         & 94.4 & $-16.7$ & $-17.6$ & $-14.7$ & $-50.0$ \\
        \bottomrule
    \end{tabular}
\end{table}

The gap between the Wald ratio and the paired estimate over adherent runs narrows as Adherence rises: 17.6 points at 47.2\%, 3.7 at 75.0\%, and at most 2.9 above 90\%, as the factor $(1-\alpha_{\mathrm{adh}})/\alpha_{\mathrm{adh}}$ of Proposition~\ref{mas:prop:sensitivity} predicts. Under the Mandatory card the non-adherent tasks (8.3\% of all tasks) show no difference and the two estimates coincide. Under the Open card the $+5.6$ ITT arises on the non-adherent tasks (52.8\%), while the adherent tasks (47.2\%) show $-5.9$; the Wald ratio of $+11.8$ therefore should not be read as the configuration-level CACE of this configuration.

These estimates are descriptive. Re-running the Single agent unchanged flips the outcome of 13.9\% of the tasks (Appendix~\ref{app:protocol-setup}), and on the adherent Open-card tasks the two arms disagree on only 5.9\% of them, so the sign of $-5.9$ is not established. The table supports the direction of the Wald--paired gap across Adherence levels, not a collaboration effect of any single configuration.

\section{Supplementary Experiments}
\label{app:part-experiments}

This part expands Section~\ref{sec:experiment}: the setup of its three experiments (\ref{app:protocol}), every result table (\ref{app:results}), how runs end and what they cost (\ref{app:runs}), traced runs (\ref{app:case-traces}), and the limitations (\ref{app:limits}). Appendix~\ref{mas:appendix} gives the identification theory behind the Wald estimate of Section~\ref{sec:oc-metrics}, and Appendix~\ref{app:part-system} describes the system.

\subsection{Experimental Setup}
\label{app:protocol}

Section~\ref{sec:experiment} reports three experiments on one runtime. The adherence experiment changes one setting of a team at a time and records whether the declared organization happened (Table~\ref{tab:five-factors}). The cross-harness comparison runs OpenCollab and four external harnesses on three benchmarks (Table~\ref{tab:benchmark-context}). The Duo--single comparison runs a team whose handoffs are fixed by code against a single agent on the same tasks. This subsection gives the task sets, the organizations, the shared settings and the metrics, and then the design of each experiment.

\paragraph{Benchmarks and task sets.} The cross-harness comparison uses 193 tasks of SWE-bench Pro \citep{deng2025swebenchpro}, 89 tasks of Terminal-Bench 2.1 \citep{merrill2026terminal} and 113 tasks of DeepSWE \citep{huang2026deepswe}. The adherence experiment uses 36 SWE-bench Pro tasks drawn at random, and every configuration runs on the same tasks. The Duo--single comparison uses the task sets of the cross-harness comparison and is paired task by task on Terminal-Bench 2.1 and DeepSWE.

A SWE-bench Pro run is resolved when every fail-to-pass test passes and no pass-to-pass test breaks, as judged by the benchmark's official evaluator. A run that produces no result, because it timed out, stalled, or its tests returned nothing, counts as unresolved. DeepSWE is graded the same way on each task's fail-to-pass and pass-to-pass tests. A Terminal-Bench 2.1 run is graded by the task's own verifier, which checks the files, program output or running service in the candidate's container; it passes when the verifier returns reward 1. When grading itself failed for an infrastructure reason (a download, a dependency, the verifier not starting), the same candidate was graded again, with no new model call.

\paragraph{Organizations.} The \emph{single agent}, OpenCollab (Base) in Table~\ref{tab:benchmark-context}, is one agent with six tools (\texttt{bash}, \texttt{file\_read}, \texttt{file\_write}, \texttt{apply\_patch}, \texttt{grep}, \texttt{git\_diff}) and no teammate. A team runs under one of two controllers (Appendix~\ref{app:arch}). In a \emph{Workflow} code issues every handoff. In a \emph{Team} the model decides whether to brief a teammate, whom, and how much to do itself. \emph{Duo}, OpenCollab (Duo) of Table~\ref{tab:benchmark-context}, is a Workflow: two Coders each produce a candidate fix and code selects one (Appendix~\ref{app:exp-duo}). The adherence experiment runs two Team rosters, named by their topology. In \emph{Star} an \emph{Adopter} receives the task, may brief two Coders, and chooses and submits the final patch. Each Coder works in its own copy of the repository and reports back with a commit. Whether the Coders run is therefore the Adopter's decision, which is what adherence records. \emph{Ring} has an \emph{analyst} that receives the task and submits, a \emph{coder} and a \emph{tester}. Every agent starts from the single agent's configuration and adds its role card. Figure~\ref{fig:topologies} draws Duo and the three Team topologies.

\paragraph{Shared settings.} In the adherence experiment every agent runs with reasoning effort at its highest setting and a streamed response, a limit of 200 steps, and the runtime's reminder to edit after repeated reads switched off. Each agent has its own token allowance of 2M unless a Budget row changes it, and the runtime refuses a call that would exceed it (Appendix~\ref{app:arch}). A run is stopped after 5{,}400\,s of wall clock with Qwen3.8-Flash and after 7{,}200\,s with DeepSeek-V4.1-Flash and GPT-5.6-Luna. In this experiment the container in which the agents run has no network access, so an agent cannot fetch a fix from outside the repository; only the runtime, outside the container, reaches the model endpoint.

\paragraph{Which runs are kept.}
\label{app:protocol-setup}
A run that ends for a reason the model caused stays in every denominator: it reached its allowance or the wall clock, produced no patch, or failed the tests. Its patch is graded as it stands. A run lost to an infrastructure fault is launched again on the same task at the same code. A run that still has no result counts as unverified for Adherence and as not resolved for Pass@1, so every configuration is read over all of its runs.

\paragraph{Metrics.} A run is \emph{adherent} when it is adherent on every one of the six axes of Table~\ref{tab:adherence} that applies to its arm. In these experiments the runtime holds four of them (role boundary, budget sharing, information flow and context policy), so only participation and delegation vary (Appendix~\ref{app:axes} gives the reasons). A Team run is therefore adherent when every teammate, every agent other than the entry agent, did work: each spent tokens and produced at least one model output. The teammates are the two Coders in Star and the coder and the tester in Ring, so the two rosters count different agents. A run is \emph{unverified} when the entry agent had already begun to delegate but the run was cut off before every teammate had worked, by the entry agent's token allowance, by the wall clock, or by a teammate that stopped without replying, or when it failed after being launched again. Such a run might have collaborated had it not been cut off, so the log cannot settle it. A run in which the entry agent never delegated is not unverified. Unverified runs count as not adherent ($A_r = 0$, Appendix~\ref{mas:setup}); counting every unverified run as adherent instead shows how much Adherence depends on this rule. Pass@1 is the share of its runs that are resolved, and Tokens is the median per run over all agents. Intervals are Clopper--Pearson 95\%. Two configurations are compared task by task on all tasks, with an exact two-sided sign test and no correction for multiple comparisons. Re-running the single agent unchanged flips the outcome of 13.9\% of the tasks, so a difference of a few tasks between two rows is within run-to-run variation.

\subsubsection{The adherence experiment}
\label{app:five-dims}

Table~\ref{tab:configs} lists every configuration of this experiment with the code used below.

\begin{table}[t]
\caption{The configurations of the adherence experiment. Each row but S$'$ is a row of Table~\ref{tab:five-factors}, or the single agent it is compared with in Table~\ref{tab:itt-pp-cace}; S$'$ reruns S unchanged. Every agent has a 2M-token allowance except in B1 and B3. Each configuration runs every task once. S and S$'$ are the single agent; Star has an Adopter and two Coders, and both rings use the Ring roster (an analyst, a coder and a tester).}
\label{tab:configs}
\centering
\footnotesize
\setlength{\tabcolsep}{4pt}
\begin{tabular}{@{}llllll@{}}
\toprule
\textbf{Code} & \textbf{Table~\ref{tab:five-factors} row} & \textbf{Card} & \textbf{Model} & \textbf{Entry tools} & \textbf{Topology} \\
\midrule
S & -- & -- & Qwen3.8-Flash & all & -- \\
S$'$ & -- & -- & Qwen3.8-Flash & all & -- \\
R & reference & Open & Qwen3.8-Flash & all & Star \\
M1 & Model & Open & DeepSeek-V4.1-Flash & all & Star \\
M2 & Model & Open & GPT-5.6-Luna & all & Star \\
T1 & Tools & Open & Qwen3.8-Flash & no edit & Star \\
T2 & Tools & Open & Qwen3.8-Flash & read-only & Star \\
B1 & Budget & Open + 0.5M stated & Qwen3.8-Flash & all & Star \\
B2 & Budget & Open + 2M stated & Qwen3.8-Flash & all & Star \\
B3 & Budget & Open + 4M stated & Qwen3.8-Flash & all & Star \\
C1 & Context & Optional & Qwen3.8-Flash & all & Star \\
C2 & Context & Mandatory & Qwen3.8-Flash & all & Star \\
G1 & Topology & Mandatory & Qwen3.8-Flash & all & Ring + shortcut \\
G2 & Topology & Mandatory & Qwen3.8-Flash & all & Ring \\
\bottomrule
\end{tabular}
\end{table}

Every row of Table~\ref{tab:five-factors} is one configuration run once on each task, twelve configurations in all. Every row except the two rings changes one setting of a single reference team: Star with Qwen3.8-Flash, the Open card, all tools, and a 2M-token allowance per agent that the card does not mention. The single agent runs the same tasks twice, the second time unchanged, to measure run-to-run variation.

\paragraph{Model.} The same team with DeepSeek-V4.1-Flash or GPT-5.6-Luna in every role. GPT-5.6-Luna is served through a different endpoint from the other two models, so its row changes the endpoint as well as the model.

\paragraph{Tools.} \emph{All tools} is the reference. \emph{No edit tools} removes the Adopter's two file-editing tools (file write and apply patch); it keeps the shell. \emph{Read-only} leaves the Adopter only its reading tools and \texttt{adopt}, which checks out a Coder's commit into the Adopter's own tree. The Coders' tools are unchanged.

\paragraph{Budget.} Each row adds two sentences to the Adopter's card and sets every agent's allowance to the stated value: ``Your budget for this run is $X$ tokens, and each Coder's is the same. On these tasks, a single agent working alone used 1,320,000 tokens on average (median 1,470,000).'' The 2M row therefore differs from the reference only in those two sentences.

\paragraph{Context: how the three cards are written.} Every Adopter card has the same body: the role, the tools, the two Coders and what each is asked for, and how a run ends. The three cards differ only in the closing section, which Table~\ref{tab:cards} gives in full. The \emph{Open} card leaves every choice to the Adopter and says that neither keeping the work nor handing it over is expected. The \emph{Mandatory} card opens by saying that the division of the work is not the Adopter's to decide, and then gives one procedure: implementation belongs to the Coders, who are asked one at a time, the second after the first has answered, and the Adopter chooses which candidate to deliver. The \emph{Optional} card gives the same procedure without that opening sentence, and adds one sentence that lets the Adopter depart from it and implement the fix itself. The three cards thus form a ladder: no instruction, an instruction the Adopter may set aside, and an instruction it may not. Because the second Coder is briefed after the first has answered, the two candidates are not independent draws. In the Optional row a Coder that stopped without replying did not notify the Adopter, so an Adopter waiting for that reply could end the run; its Adherence and Pass@1 are therefore not strictly comparable with the other rows.

\begin{table}[t]
\caption{The closing section of the Adopter's card, the only text that differs between the Open, Optional and Mandatory cards; the rest of the card is identical. Text the Optional and Mandatory cards share is in gray; what distinguishes each card is in black.}
\label{tab:cards}
\centering
\footnotesize
\begin{tabularx}{\textwidth}{@{}lX@{}}
\toprule
\textbf{Card} & \textbf{Closing section} \\
\midrule
Open & \textbf{What is yours to judge.} Nothing above tells you what order to do
things in, whom to talk to, or how much of the work to do yourself. Keeping all
of it and handing parts of it over are both open to you, and neither one is
what you are expected to do. Those are your calls, and you make them the way
you would judge any piece of work: what the request actually needs, what is
worth another agent's attention, what it costs to describe a piece of it well
enough to hand over, and what it costs to carry all of it in one budget and one
conversation. \\
\midrule
Optional & \textcolor{gray}{\textbf{What this run asks of you.} Implementation is the Coders'. There are two of them and neither can see what the other is doing, so ask them one at a time: send the first a message with \texttt{message\_agent} describing the problem and let it work out its own fix, and when its answer is back, write the second its own brief and let it work out its own. What goes into that second brief is yours -- what you learned from the first is yours to use or to leave out. Which of the two this run delivers is yours to choose.} \textbf{You may depart from this and do the
implementation yourself; that is your call, and there are runs where it is the
better one.} \textcolor{gray}{Use \texttt{team\_status} first if you want to see who is live. Everything else about how you work -- what you read, how you diagnose, what you put in each message -- is still yours.} \\
\midrule
Mandatory & \textcolor{gray}{\textbf{What this run asks of you.}} \textbf{This run is not asking you to decide how to
divide the work. It is asking you to divide it a particular way, so do that.}
\textcolor{gray}{Implementation is the Coders'. There are two of them and neither can see what the other is doing, so ask them one at a time: send the first a message with \texttt{message\_agent} describing the problem and let it work out its own fix, and when its answer is back, write the second its own brief and let it work out its own. What goes into that second brief is yours -- what you learned from the first is yours to use or to leave out. Which of the two this run delivers is yours to choose. Use \texttt{team\_status} first if you want to see who is live. Everything else about how you work -- what you read, how you diagnose, what you put in each message -- is still yours.} \\
\bottomrule
\end{tabularx}
\end{table}

\paragraph{Topology.} The three rows differ in who can message whom (Figure~\ref{fig:topologies}). \emph{Star} is the reference team: the Adopter and each Coder can message each other, and the two Coders cannot message each other. The two rings use the Ring roster, each agent with the reference's tools. The analyst's card opens with the same two sentences as the Mandatory card and then asks the analyst to leave the implementation to the coder. In the \emph{Ring}, the analyst can message only the coder, the coder only the tester and the tester only the analyst. \emph{Ring + shortcut} adds one edge, from the coder to the analyst, so the coder may report directly and skip the tester.

\begin{figure}[h]
\centering
\tikzset{topo/.style={font=\footnotesize,inner sep=2pt,line width=0.5pt,
  >={Stealth[length=4pt,width=3.5pt]}}}
\resizebox{\linewidth}{!}{%
\begin{tikzpicture}[topo]
\node(a) at (1.2,0){code}; \node(x) at (0,-1){Coder A}; \node(y) at (2.4,-1){Coder B};
\node(j) at (3.9,0){Adjudicator};
\draw[<->](a)--(x); \draw[<->](a)--(y); \draw[<->,dashed](a)--(j);
\node at (1.2,-1.7){Duo, Workflow};
\end{tikzpicture}\hspace{1cm}
\begin{tikzpicture}[topo]
\node(a) at (1.2,0){Adopter}; \node(x) at (0,-1){Coder A}; \node(y) at (2.4,-1){Coder B};
\draw[<->](a)--(x); \draw[<->](a)--(y);
\node at (1.2,-1.7){Star};
\end{tikzpicture}\hspace{1cm}
\begin{tikzpicture}[topo]
\node(n) at (1.2,0){analyst}; \node(c) at (0,-1){coder}; \node(t) at (2.4,-1){tester};
\draw[<->](n)--(c); \draw[->](c)--(t); \draw[->](t)--(n);
\node at (1.2,-1.7){Ring + shortcut};
\end{tikzpicture}\hspace{1cm}
\begin{tikzpicture}[topo]
\node(n) at (1.2,0){analyst}; \node(c) at (0,-1){coder}; \node(t) at (2.4,-1){tester};
\draw[->](n)--(c); \draw[->](c)--(t); \draw[->](t)--(n);
\node at (1.2,-1.7){Ring};
\end{tikzpicture}}
\caption{The four organizations. Left: Duo, a Workflow, OpenCollab (Duo) of Table~\ref{tab:benchmark-context}; every arrow is a handoff that code makes on every run, and the dashed one runs only when the selection rule does not decide. The other three are the topologies of Table~\ref{tab:five-factors}, run as Teams: an arrow means the agent at its tail can message the agent at its head, and whether it does is the model's choice.}
\label{fig:topologies}
\end{figure}
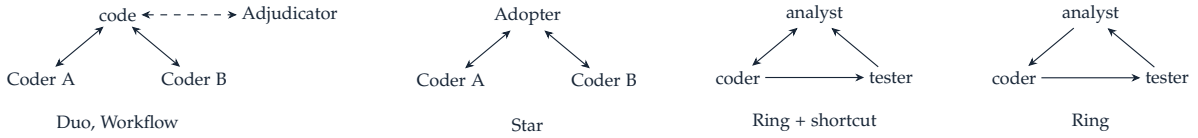

\subsubsection{The cross-harness comparison}
\label{app:exp-harness}

This comparison asks whether OpenCollab is a competitive harness, so that what the adherence experiment finds is not a property of a weak one. Six harnesses run on the three benchmarks, all with GPT-5.6-Luna: mini-swe-agent, Codex CLI, Claude Code, OpenHands, OpenCollab (Base) and OpenCollab (Duo).

\paragraph{Network access and answer leakage.} Answer leakage is the main risk of letting an agent reach the internet, so we audited for it rather than relying on a network-off setting. On Terminal-Bench 2.1 the runner follows each task's own \texttt{allow\_internet} flag, because some tasks need it to install dependencies. We searched every trajectory of every harness for fetches of the benchmark's repository, \texttt{solution} and \texttt{tests} directories, reference implementations and reference values, and re-read each hit: the tool call, what it returned, what the model did next, and whether the resulting candidate was selected. An attempt counts only if the model received the material. Most affected tasks were run again in a fresh session from the original task statement and base image, with the old candidate, session and grading feedback kept out, and with a filter on the source and content of what a tool returns. Every pass obtained with leaked material was withdrawn or replaced before Table~\ref{tab:benchmark-context} was computed.

\subsubsection{The Duo--single comparison}
\label{app:exp-duo}

The adherence experiment shows that a Team's handoffs may not happen, so a difference in Pass@1 between a Team and a single agent cannot be read as the effect of the organization. OpenCollab (Duo) is a Workflow: code issues every handoff, so both of its Coders run on every task.

\paragraph{The workflow.} OpenCollab (Duo) has two solving agents and a judge that is called only when needed. \emph{Coder A} looks for the narrowest root cause and writes a minimal complete fix. \emph{Coder B} works from the same original task and aims at completeness across components, APIs, lifecycles and boundary conditions; it may receive the public test commands that Coder A ran. Each works in its own candidate environment, with the single agent's configuration and six tools. The code then selects between the two candidates by a fixed rule; when the rule does not decide, it calls a read-only \emph{Adjudicator}, which compares the candidates' evidence with the public requirements of the task (on Terminal-Bench 2.1 it also has a read-only tool that pages through that evidence). The code, not an agent, applies the selected candidate.

\paragraph{Paired results.} Table~\ref{tab:duo-pairs} pairs OpenCollab (Duo) with OpenCollab (Base) task by task, with an exact two-sided sign test on the tasks where only one of them passed.

\begin{table}[t]
\caption{OpenCollab (Duo) against OpenCollab (Base), paired by task, as a share of each benchmark's tasks (\%). $p$ is the exact two-sided sign test on the two middle columns.}
\label{tab:duo-pairs}
\centering
\footnotesize
\begin{tabular}{@{}lrrrrr@{}}
\toprule
\textbf{Benchmark} & \textbf{Both pass} & \textbf{Only Duo} & \textbf{Only Base} & \textbf{Both fail} & $p$ \\
\midrule
Terminal-Bench 2.1 & 73.03 & 10.11 & 6.74 & 10.11 & 0.61 \\
DeepSWE & 46.90 & 18.58 & 8.85 & 25.66 & 0.071 \\
\bottomrule
\end{tabular}
\end{table}

\subsection{Full Results}
\label{app:results}
\label{app:results-tables}

Table~\ref{tab:stats} gives, for every configuration of Table~\ref{tab:configs}, the rates behind Table~\ref{tab:five-factors} and the task-by-task comparisons behind its reading. A comparison $a{:}b$ gives the share of tasks on which only this configuration succeeded ($a$) and the share on which only the other one did ($b$), in percent of all tasks, with the exact two-sided sign test in parentheses. The reference is R for every Star row and G1 for G2. C1 is not compared, for the reason given under Context in Appendix~\ref{app:five-dims}. The single-agent column is left empty for M1 and M2, whose models have no single-agent run.

\begin{table}[t]
\caption{Adherence and Pass@1 of every configuration, with paired comparisons (share of tasks on which only this configuration succeeded : share on which only the other did, in \%, exact sign test $p$). Adherence for a Star run counts both Coders; for a Ring run, the coder and the tester.}
\label{tab:stats}
\centering
\footnotesize
\setlength{\tabcolsep}{4pt}
\begin{tabular}{@{}lcccccc@{}}
\toprule
& \multicolumn{3}{c}{\textbf{Adherence}} & \multicolumn{3}{c}{\textbf{Pass@1}} \\
\cmidrule(lr){2-4}\cmidrule(lr){5-7}
\textbf{Code} & \textbf{\% [95\% CI]} & \textbf{Unverified \%} & \textbf{vs reference} & \textbf{\%} & \textbf{vs reference} & \textbf{vs S} \\
\midrule
S & -- & -- & -- & 69.4 & -- & -- \\
S$'$ & -- & -- & -- & 66.7 & 5.6:8.3 (1.00) & -- \\
R & 47.2 [30.4, 64.5] & 13.9 & -- & 75.0 & -- & 8.3:2.8 (0.63) \\
M1 & 66.7 [49.0, 81.4] & 22.2 & 36.1:16.7 (0.17) & 72.2 & 5.6:8.3 (1.00) & -- \\
M2 & 97.2 [85.5, 99.9] & 2.8 & 52.8:2.8 ($<$0.001) & 61.1 & 5.6:19.4 (0.18) & -- \\
T1 & 86.1 [70.5, 95.3] & 0.0 & 41.7:2.8 ($<$0.001) & 66.7 & 2.8:11.1 (0.38) & 2.8:5.6 (1.00) \\
T2 & 94.4 [81.3, 99.3] & 2.8 & 50.0:2.8 ($<$0.001) & 52.8 & 0.0:22.2 (0.008) & 5.6:22.2 (0.11) \\
B1 & 47.2 [30.4, 64.5] & 2.8 & 25.0:25.0 (1.00) & 44.4 & 0.0:30.6 (0.001) & 0.0:25.0 (0.004) \\
B2 & 75.0 [57.8, 87.9] & 8.3 & 38.9:11.1 (0.03) & 69.4 & 2.8:8.3 (0.63) & 8.3:8.3 (1.00) \\
B3 & 77.8 [60.8, 89.9] & 2.8 & 36.1:5.6 (0.007) & 75.0 & 2.8:2.8 (1.00) & 8.3:2.8 (0.63) \\
C1 & 63.9 [46.2, 79.2] & 16.7 & -- & 61.1 & -- & -- \\
C2 & 91.7 [77.5, 98.2] & 8.3 & 52.8:8.3 ($<$0.001) & 72.2 & 2.8:5.6 (1.00) & 2.8:0.0 (1.00) \\
G1 & 66.7 [49.0, 81.4] & 13.9 & -- & 66.7 & -- & 5.6:8.3 (1.00) \\
G2 & 91.7 [77.5, 98.2] & 8.3 & 27.8:2.8 (0.01) & 63.9 & 8.3:11.1 (1.00) & 11.1:16.7 (0.75) \\
\bottomrule
\end{tabular}
\end{table}

\subsection{Adherence Across Task Counts}
\label{app:scaling}

We extend the adherence experiment to twice the tasks of Table~\ref{tab:stats}, drawn at random, and from them draw random subsets of 9, 18 and 36 tasks. Figure~\ref{fig:scaling} plots every team configuration's Adherence on each, one panel per dimension of Table~\ref{tab:five-factors}, and Table~\ref{tab:stats-x72} gives the rates of all configurations on that larger set. The trends agree across task counts: the variants that Table~\ref{tab:five-factors} reads as far above the reference (GPT-5.6-Luna, No edit tools, Read-only, Mandatory card, Ring) stay above it at every task count, from 9 tasks on. The exception is Budget, where 2M and 4M per role swap order as the task count grows: with the larger allowance, the Adopter tends to keep working on the task itself until the wall clock, and one of the Coders never gets a turn.

\begin{figure}[!htbp]
\centering
\includegraphics[width=\linewidth]{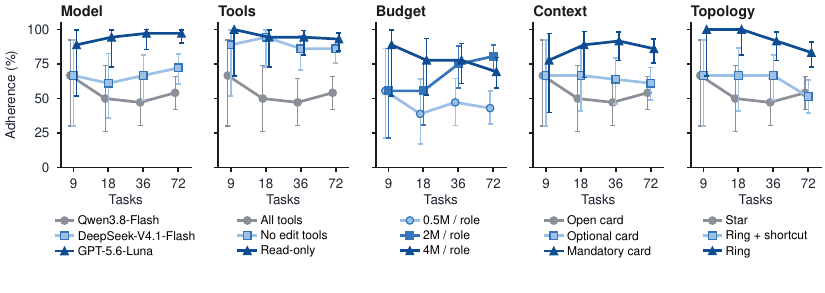}
\caption{Adherence on random subsets of 9, 18, 36, and 72 tasks, with 95\% Clopper--Pearson intervals, one panel per dimension with the variants of Table~\ref{tab:five-factors}. Grey is the reference configuration; the three Budget variants are all changed from it.}
\label{fig:scaling}
\end{figure}

\begin{table}[t]
\caption{Adherence and Pass@1 of every configuration on twice the tasks of Table~\ref{tab:stats}, with paired comparisons (share of tasks on which only this configuration succeeded : share on which only the other did, in \%, exact sign test $p$). Columns as in Table~\ref{tab:stats}. Adherence for a Star run counts both Coders; for a Ring run, the coder and the tester.}
\label{tab:stats-x72}
\centering
\footnotesize
\setlength{\tabcolsep}{4pt}
\begin{tabular}{@{}lcccccc@{}}
\toprule
& \multicolumn{3}{c}{\textbf{Adherence}} & \multicolumn{3}{c}{\textbf{Pass@1}} \\
\cmidrule(lr){2-4}\cmidrule(lr){5-7}
\textbf{Code} & \textbf{\% [95\% CI]} & \textbf{Unverified \%} & \textbf{vs reference} & \textbf{\%} & \textbf{vs reference} & \textbf{vs S} \\
\midrule
S & -- & -- & -- & 55.6 & -- & -- \\
S$'$ & -- & -- & -- & 56.9 & 6.9:5.6 (1.00) & -- \\
R & 54.2 [42.0, 66.0] & 12.5 & -- & 63.9 & -- & 11.1:2.8 (0.11) \\
M1 & 72.2 [60.4, 82.1] & 16.7 & 31.9:13.9 (0.04) & 61.1 & 4.2:6.9 (0.73) & -- \\
M2 & 97.2 [90.3, 99.7] & 2.8 & 45.8:2.8 ($<$0.001) & 52.8 & 5.6:16.7 (0.08) & -- \\
T1 & 86.1 [75.9, 93.1] & 4.2 & 36.1:4.2 ($<$0.001) & 58.3 & 2.8:8.3 (0.29) & 6.9:4.2 (0.73) \\
T2 & 93.1 [84.5, 97.7] & 5.6 & 41.7:2.8 ($<$0.001) & 44.4 & 1.4:20.8 ($<$0.001) & 4.2:15.3 (0.06) \\
B1 & 43.1 [31.4, 55.3] & 5.6 & 18.1:29.2 (0.23) & 41.7 & 0.0:22.2 ($<$0.001) & 2.8:16.7 (0.01) \\
B2 & 80.6 [69.5, 88.9] & 5.6 & 33.3:6.9 ($<$0.001) & 58.3 & 4.2:9.7 (0.34) & 8.3:5.6 (0.75) \\
B3 & 69.4 [57.5, 79.8] & 13.9 & 25.0:9.7 (0.04) & 62.5 & 4.2:5.6 (1.00) & 8.3:1.4 (0.13) \\
C1 & 61.1 [48.9, 72.4] & 23.6 & -- & 51.4 & -- & -- \\
C2 & 86.1 [75.9, 93.1] & 12.5 & 38.9:6.9 ($<$0.001) & 58.3 & 1.4:6.9 (0.22) & 4.2:1.4 (0.63) \\
G1 & 51.4 [39.3, 63.3] & 22.2 & -- & 56.9 & -- & 6.9:5.6 (1.00) \\
G2 & 83.3 [72.7, 91.1] & 6.9 & 34.7:2.8 ($<$0.001) & 51.4 & 4.2:9.7 (0.34) & 6.9:11.1 (0.58) \\
\bottomrule
\end{tabular}
\end{table}

\subsection{How Runs End and What They Cost}
\label{app:runs}

This section covers the Qwen3.8-Flash configurations of Table~\ref{tab:configs}, that is, every configuration except M1 and M2. A run ends in one of three ways (Figure~\ref{fig:runs-ends}). It completes when the entry agent (the Adopter, or the analyst in a Ring) submits. It is stopped for its token budget when the entry agent has used up its own allowance before submitting. It is stopped for wall clock when it reaches the 5{,}400\,s limit. Budget stops dominate only where the allowance is smallest: in B1 92\% of runs stop for budget and 8\% complete. With the largest allowance, B3, a third of the runs reach the wall clock instead. T2 completes most often, in 89\% of its runs.

Figures~\ref{fig:runs-tokens} and~\ref{fig:runs-duration} show what a run costs in tokens and in time, and how a team run's cost splits across its three agents. A single-agent run uses 1.3M to 1.4M tokens on average against an allowance of 2M. A team run uses 2.9M to 4.2M against 6M, except in the two Budget rows whose allowance differs: B1 uses 0.9M of 1.5M and B3 4.7M of 12M. The entry agent takes 33\% to 50\% of a team's tokens, except in T2, where it can only read and adopt a commit and takes 17\%. Because exactly one agent runs at any moment (Appendix~\ref{app:limits}), a run's wall clock divides into the time each agent spends in its own model calls and tool calls, plus the harness's time between turns, which stays below 2\%. The time split follows the token split: the entry agent's share ranges from 27\% in T2 to 59\% in B1. A single-agent run takes 28 to 30 minutes on average and a team run 44 to 59 minutes, except in B1 (20 minutes).

Figure~\ref{fig:runs-seat-time} shows when each agent spends its tokens. Each run is cut into ten equal slices of time, from its start to its last model call, and the tokens in each slice are divided among the agents. The entry agent spends nearly all tokens in the first slices, the other two agents take over in the middle, and in every configuration the entry agent's share rises again in the last slice. In B1 the entry agent spends all the tokens of the first three slices; in T2 it spends almost nothing in the middle of the run.

\begin{figure}[!htbp]
\centering
\includegraphics[width=\linewidth]{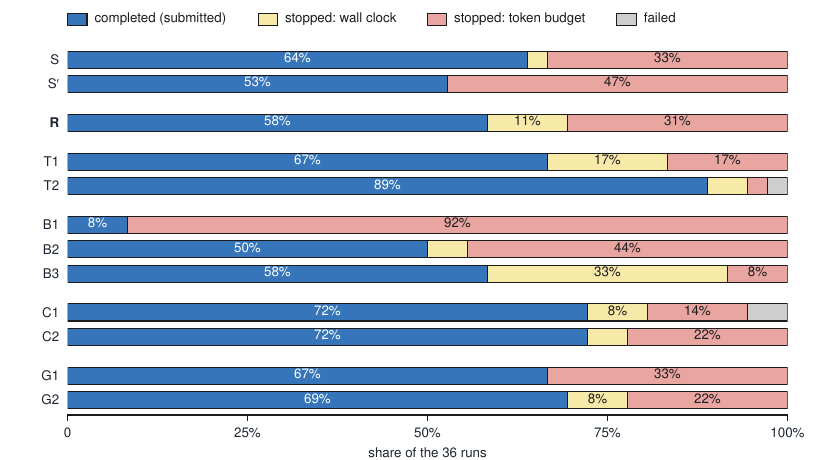}
\caption{How runs end, as a share of each configuration's runs (Appendix~\ref{app:protocol-setup}). Completed: the entry agent submitted. Token budget: the entry agent used up its allowance before submitting. Wall clock: the run reached 5{,}400\,s. Failed: no result after the run was launched again. Configuration codes as in Table~\ref{tab:configs}.}
\label{fig:runs-ends}
\end{figure}

\begin{figure}[!htbp]
\centering
\includegraphics[width=\linewidth]{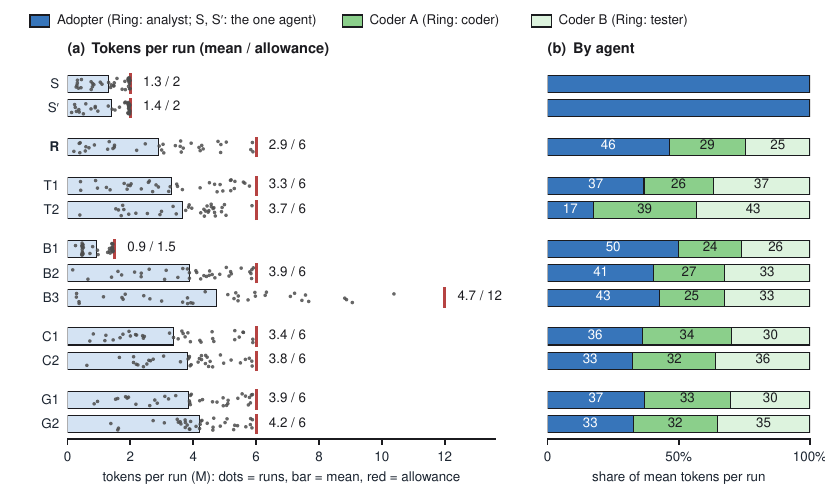}
\caption{Tokens per run. (a) Each dot is a run; the bar is the mean and the red tick is the run's allowance (each agent's allowance times the number of agents), with ``mean / allowance'' in millions at the right. (b) The mean split across agents, in percent. In a Ring the three agents are the analyst, the coder and the tester.}
\label{fig:runs-tokens}
\end{figure}

\begin{figure}[!htbp]
\centering
\includegraphics[width=\linewidth]{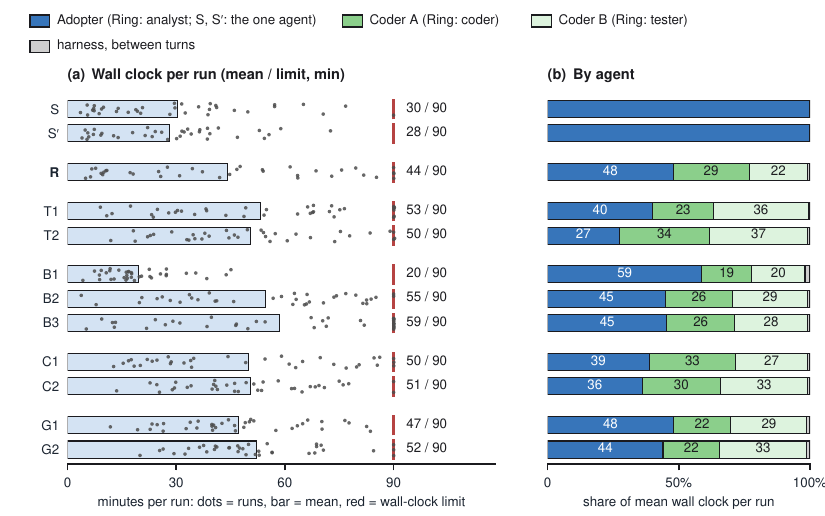}
\caption{Wall clock per run. (a) Each dot is a run; the bar is the mean and the red tick is the 5{,}400\,s limit, with ``mean / limit'' in minutes at the right. (b) The mean split across agents, where an agent's time is the duration of its own model calls and tool calls; grey is the harness's time between turns, a thin sliver at the right of each bar.}
\label{fig:runs-duration}
\end{figure}

\begin{figure}[!htbp]
\centering
\includegraphics[width=\linewidth]{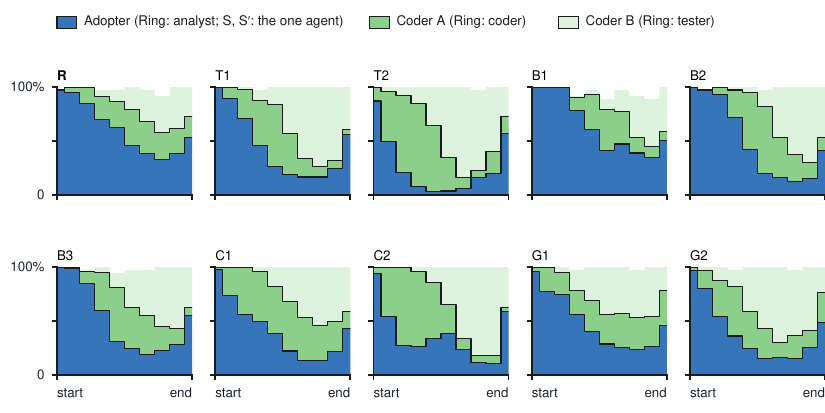}
\caption{Share of tokens by agent in each tenth of a run, from its start to its last model call, averaged over the configuration's runs. One panel per team configuration of Table~\ref{tab:configs} with Qwen3.8-Flash.}
\label{fig:runs-seat-time}
\end{figure}

\subsection{Three Runs, Traced}
\label{app:case-traces}

Figures~\ref{fig:case-trace-1}--\ref{fig:case-trace-3} draw three runs of the adopter--two-coder team from their event streams. Each lane is one agent; an outlined bar is one turn, and a block inside it is one model call, coloured by what the call did. Only one agent runs at a time, so a message to another agent is drawn as a vertical arrow at the moment it is sent, and the dotted line after it is the time the woken agent waits before its turn starts. Labelled points mark moments checked against the trajectory, and the bottom row of each panel is the code that would be graded if the run stopped at that moment; the legend sits above Figure~\ref{fig:case-trace-1}, and all three share one page. Each is a single run chosen to show one way the declared organization is or is not realized, so none of them estimates how often that happens, and their time axes differ.

\makeatletter\setlength{\@fptop}{0pt}\makeatother
\begin{figure}[p]
	\centering
	\includegraphics[width=\textwidth]{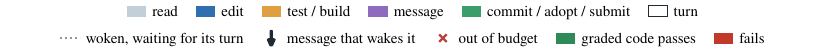}\\[6pt]
	\includegraphics[width=\textwidth]{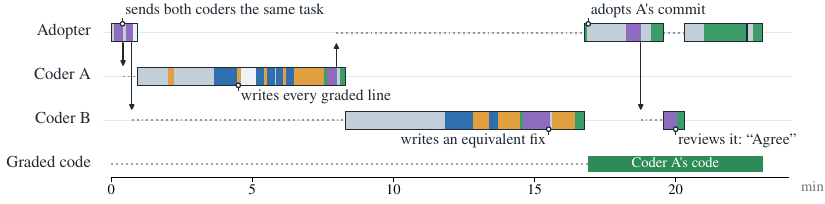}
	\caption{Read-only variant (the adopter can only read files and adopt a commit), resolved: two coders fix the task and the adopter adopts Coder A's commit. A SWE-bench Pro task from the teleport repository.}
	\label{fig:case-trace-1}
	\vspace{6pt}
	\includegraphics[width=\textwidth]{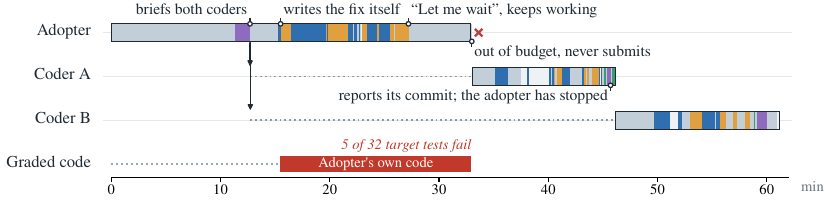}
	\caption{Reference team, unresolved: the adopter briefs both coders but keeps working until its budget ends; the coders run only after its turn ends. A SWE-bench Pro task from the flipt repository.}
	\label{fig:case-trace-2}
	\vspace{6pt}
	\includegraphics[width=\textwidth]{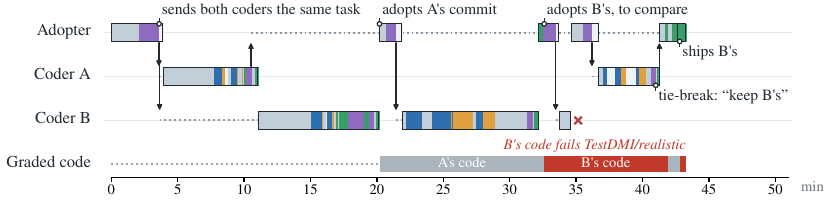}
	\caption{Read-only variant, unresolved: cross-review ships the commit that fails. A SWE-bench Pro task from the teleport repository.}
	\label{fig:case-trace-3}
\end{figure}

\paragraph{Delegation that works (Figure~\ref{fig:case-trace-1}).} In the Read-only variant the adopter can read files and adopt a commit, but cannot edit. It sends the same task to both coders. Coder A writes the fix, and every line of the graded patch comes from that edit; Coder B independently writes an equivalent one. The adopter adopts Coder A's commit and asks Coder B to review it, and Coder B agrees.

\paragraph{Briefed but not waited for (Figure~\ref{fig:case-trace-2}).} In the reference team the adopter briefs both coders at minute 12.7. The tool that sends the brief replies that the teammate runs on its next turn and that ending the turn is how the adopter waits. The adopter instead implements the change itself; at minute 27.2 it writes that it will wait for the coders' replies, then makes five more model calls. Its budget runs out at minute 33 before it submits, and only then do the coders run. Coder A reports its commit to the adopter, which has already stopped. The five target tests the graded code fails reference a constant name that the task statement does not give, so this failure does not follow from the missing delegation.

\paragraph{Cross-review that picks the failing commit (Figure~\ref{fig:case-trace-3}).} Again in the Read-only variant, the adopter adopts Coder A's commit, then Coder B's to compare them. Coder A reports that Coder B's commit passes all fifteen of its probes and recommends keeping it, and the adopter ships it. Coder B's code reads a file in a way that bypasses the permission error that a target test injects, and fails that test.

The figures show forms the organization takes; they are not a comparison with the single agent.

\subsection{Limitations}
\label{app:limits}

Each limitation below narrows how far a number in this paper reaches.

\paragraph{Three models, not a model sweep.} The three model snapshots of Table~\ref{tab:five-factors} show that Adherence moves when the model changes, but three snapshots cannot estimate how it varies across models in general.

\paragraph{One agent at a time.} All agents in a team share one lock, so exactly one agent runs at any moment; a message to a busy agent is delivered when its turn ends. Organizations whose value depends on agents working truly in parallel are therefore outside what these experiments measure.

\paragraph{One workflow, not the space of workflows.} OpenCollab (Duo) is a single decomposition of the task. A different script could deliver a different organization just as faithfully and perform differently.

\section{System Details}
\label{app:part-system}

This part describes the code behind Sections~\ref{sec:oc-architecture}--\ref{sec:oc-metrics}: the runtime and its two controllers (\ref{app:arch}), how the code is organized (\ref{app:code}), how a collaboration is declared (\ref{app:declare}), an audit of ten agent artifacts, with OpenCollab scored by the same rule, against the conditions of Section~\ref{sec:oc-controlled} (\ref{app:matrix}), the run record (\ref{app:schema}), and how adherence is computed from it (\ref{app:axes}).

\subsection{Runtime and Control Regimes}
\label{app:arch}

Sections~\ref{sec:oc-architecture} and~\ref{sec:oc-declare} describe one runtime that carries three control regimes: a Team, a script-ordered Workflow, and a single agent. This appendix gives the limits every session runs under, its context compaction, the two controllers above it, the workspaces, and the token budget.

\paragraph{The session loop.} Every session, whichever controller opened it, runs the same loop, a state machine with ten states. \texttt{IDLE} precedes the first step. \texttt{PRECHECK} admits or refuses the next model call, \texttt{CALLING\_LLM} makes it, and \texttt{HANDLING\_RESPONSE} sends a reply with tool calls to \texttt{EXECUTING\_TOOLS}; \texttt{AUTOSAVING} records the step and returns to \texttt{PRECHECK}. A tool whose result is deferred parks the session in \texttt{AWAITING\_EVENTS} until the controller resumes it. \texttt{DONE}, \texttt{STOPPED} and \texttt{ERROR} are the end states (Table~\ref{tab:termination}). A transition that is not in the table of legal transitions raises an error instead of being taken (Figure~\ref{fig:lifecycle}).

\begin{figure}[!htb]
  \centering
  \includegraphics[width=\textwidth]{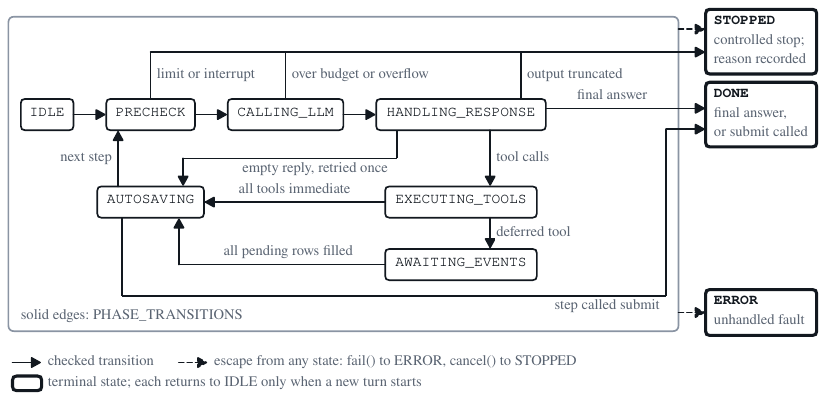}
  \caption{The session state machine that every session runs. Solid edges are the checked transition table; a transition not in the table raises an error instead of being taken. Two escapes bypass the table from any state: an unhandled fault ends the session in \texttt{ERROR}, and a cancellation from outside the loop ends it in \texttt{STOPPED}. \texttt{STOPPED} is a controlled stop at a limit or a refusal, with the reason recorded (Table~\ref{tab:termination}), not a fault. Each end state returns to \texttt{IDLE} only when a new turn starts.}
  \label{fig:lifecycle}
\end{figure}

\paragraph{Limits.} Every limit is checked before a call is made. \texttt{PRECHECK} stops the session once it has reached its step limit, its limit on repeated tool calls, or its token allowance. In \texttt{CALLING\_LLM} the call is then priced with a conservative estimate of its input, and refused if it would not fit in what remains of the allowance. One check also runs after the call, because a call can still cost more than its estimate; a role that ends up past its allowance in this way stops at once. A role stopped at any of these points ends in \texttt{STOPPED} with the reason recorded (Table~\ref{tab:termination}). Because an organization decides only which sessions exist and when they start, never what happens inside one, these limits bind identically whichever controller opened the session.

\paragraph{Context compaction.} Before each model call the session compacts a copy of its history in four stages applied in order. The first caps every tool result at 16,000 characters and runs on every call. The other three act only when the estimated history exceeds a trigger set 33,000 tokens below the model's context window. They clear the content of old tool results while keeping the five most recent verbatim, then drop whole old tool exchanges, oldest first, and last replace the remaining old span with a summary written by a separate model call. Each stage that fires writes a record to the event log. If the provider still rejects the prompt as too long, a forced pass ignores the trigger and the call is retried once; a second rejection stops the session. The stored history is never altered by compaction.

\paragraph{Two controllers.} Above the loop, a controller decides which sessions exist, when each one starts, and how their outputs travel (Figure~\ref{fig:control-runtime}). The Team controller opens a session for every declared role and leaves it to the model whether to send work along the declared edges. The Workflow controller is a Python script that opens sessions in a fixed order and passes their outputs on itself. A single-agent run is one session on the same runtime. Table~\ref{tab:arms} lists the three arms built from these regimes.

\begin{figure}[t]
\setlength{\belowcaptionskip}{0pt}
  \centering
  \includegraphics[width=\textwidth]{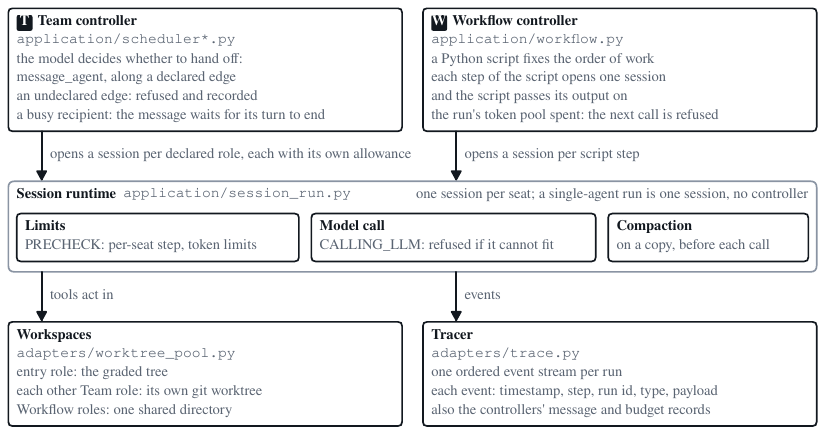}
  \caption{Two controllers over one session runtime. The Team controller opens a session for every declared role, each with its own token allowance; the model decides whether to hand work on with \texttt{message\_agent}, and a message along an undeclared edge is refused and recorded. The Workflow controller is a Python script that fixes the order of work and opens one session per step; a single-agent run is one session with no controller above it. All sessions run the same code, which checks every limit before a model call, acts in the role's workspace, and writes to one ordered event stream per run (Appendix~\ref{app:schema}).}
  \label{fig:control-runtime}
\end{figure}

\begin{table}[t]
\caption{The three arms. The last column says how a run is known to have realized its arm. Only in the Team is it a measurement of the model's choice. In the Workflow the script, not the model, issues every handoff it reaches.}
\label{tab:design}
\label{tab:arms}
\centering
\footnotesize
\begin{tabularx}{\textwidth}{lYll}
\toprule
\textbf{Arm} & \textbf{Definition} & \textbf{Differs from Single in} &
\textbf{Realized} \\
\midrule
Single & one agent works the task to completion & --- & trivially \\
Workflow & shared role cards; a script decides the order of work, what each
role receives, and when to stop & organization, enforced & by code \\
Team & shared role cards; the model decides whether to delegate, what to
send, and when to stop & organization, offered & measured,
$\widehat{\alpha}_{\mathrm{adh}}$ \\
\bottomrule
\end{tabularx}
\end{table}

\paragraph{Workspaces.} The entry role is the role that receives the task. Its workspace is what the benchmark grades when the run ends, and we call that workspace the graded tree. The Team controller builds every declared teammate before the first model call, in its own git worktree over the same object store, and gives it no task: a teammate starts working only when a message reaches it.

\paragraph{The token budget.} Each role has its own fixed token allowance, set before the run. Building a teammate takes nothing from the team's tokens, so a teammate that is never used spends nothing. Grading runs after the run has ended and is not charged to any role's allowance.

\subsection{How the Code Is Organized}
\label{app:code}

The system is two code bases. The runtime, the Python package \texttt{opencollab}, runs the sessions and the two controllers of Appendix~\ref{app:arch}. The evaluator, the package \texttt{opencollab\_eval}, defines the arms, drives the batches and prepares each run's patch for grading. Every path and count below is taken from the code of the adherence experiment.

\paragraph{Four layers in the runtime.} Of the runtime's 167 Python files, 158 fall into four layers (Table~\ref{tab:code-layers}), and apart from five import statements in the command-line interface, described below, no module imports from a layer above its own. The domain layer holds the session state machine and the other plain data types, and it uses nothing outside the Python standard library. The application layer holds the session loop with its limits, the two controllers and context compaction. The adapters layer holds the code that reaches outside the process: model providers, tools, workspaces and git worktrees, and the event log. The bootstrap layer reads a team file, resolves its tool names to implementations, and assembles the other three layers. Above the layers sits a public surface of nine files: the client, whose \texttt{agent}, \texttt{team} and \texttt{workflow} calls start the three regimes, and the modules for tools, environments, team files and workflows. Figure~\ref{fig:clean-arch} draws the four layers as the rings of a clean architecture; the layer contract in the code carries that name.

\begin{figure}[!htb]
  \centering
  \includegraphics[width=\textwidth]{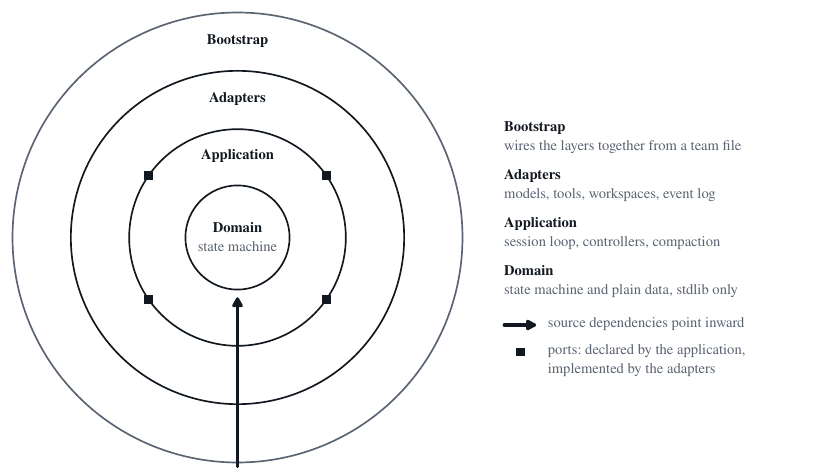}
  \caption{The runtime's four layers as a clean architecture. Every import points inward, from bootstrap to adapters to application to domain; the application declares ports, interfaces that the adapters implement, so it never imports a model provider, a tool or a workspace directly. Figure~\ref{fig:code-layers} gives the files in each layer.}
  \label{fig:clean-arch}
\end{figure}

\begin{table}[t]
\centering
\small
\caption{The runtime's layers, top to bottom. Apart from five statements in the command-line interface (see text), no module imports from a layer above its own. \emph{Lines} counts every line of the layer's Python files.}
\label{tab:code-layers}
\begin{tabular}{@{}l r r >{\raggedright\arraybackslash}p{0.60\textwidth}@{}}
\toprule
\textbf{Layer} & \textbf{Files} & \textbf{Lines} & \textbf{What it holds} \\
\midrule
Public surface & 9 & 894 & The client (\texttt{sdk/client.py}); the modules
for tools, environments, team files and workflows \\
\addlinespace
Bootstrap & 22 & 6,690 & Loading a team file (\texttt{team\_config.py},
Appendix~\ref{app:declare}); resolving tool names (\texttt{tool\_registry.py});
assembling a session, a team or a workflow run \\
\addlinespace
Adapters & 72 & 16,171 & Model providers and the per-model capability table
(\texttt{llm/}); the tools, among them the patch tool, the shell and
\texttt{message\_agent} (\texttt{tools/}); local, container and worktree
workspaces; the sandbox policy (\texttt{safety.py}); the per-run event stream
(\texttt{trace.py}, Appendix~\ref{app:schema}) \\
\addlinespace
Application & 51 & 15,202 & The session loop and its limits
(\texttt{session\_run.py}); the Team controller
(\texttt{application/scheduler.py} and its \texttt{scheduler\_*} and
\texttt{\_scheduler\_*} modules); the Workflow controller (\texttt{workflow.py}
and its \texttt{workflow\_*} modules); context compaction (\texttt{shaping/}) \\
\addlinespace
Domain & 13 & 1,820 & The ten-state session machine and its legal
transitions (\texttt{session.py}, Appendix~\ref{app:arch}); the team topology
(\texttt{team.py}); the table of sessions that the Team controller keeps, as
plain data (\texttt{domain/scheduler.py}); agents, events and token estimates \\
\bottomrule
\end{tabular}
\end{table}

Figure~\ref{fig:code-layers} draws the same layers with the modules this appendix names.

\begin{figure}[t]
  \centering
  \includegraphics[width=\textwidth]{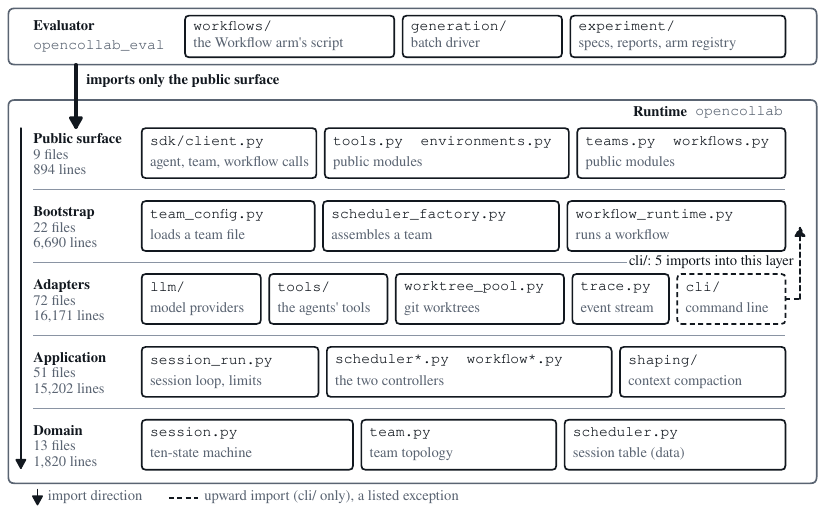}
  \caption{The two code bases. The runtime's layers run top to bottom, each with its file and line counts (Table~\ref{tab:code-layers}), and a module imports only from its own layer or the ones below; the one exception is five imports from the command-line interface into the bootstrap layer (dashed), which the experiments do not use. The evaluator reaches the runtime only through its public surface.}
  \label{fig:code-layers}
\end{figure}

\paragraph{The layering is checked.} Four import contracts, run in continuous integration, state the rules: the order of the layers; that the domain imports none of the runtime's third-party dependencies; that the client and the modules for tools, environments and workflows never import the adapters or the domain directly; and that no two modules of the application, adapters or domain layer import each other in a cycle, except one known cycle between two modules of the Workflow controller. All four hold. We confirmed this by running the contracts on the code of the experiments, and confirmed that the check can fail by adding two forbidden imports to the domain, which broke exactly the two contracts they violate. Counting import statements directly agrees: every import in the four layers points into its own layer or a lower one, except five statements in three files of the command-line interface, which belongs to the adapters layer, that reach up into the bootstrap layer. The contract file lists these as known exceptions. The experiments never start the runtime through the command-line interface.

\paragraph{The evaluator.} The evaluator's modules import the runtime only through its public surface. An arm is therefore defined by what it passes through that surface, namely role prompts, tool bundles, and a team file or a script, while the session loop and the code that enforces its limits are the runtime's own in every arm. The team files are kept in the runtime's repository, one per configuration. The batch driver (\texttt{generation/gen\_prediction\_batch.py}) runs every arm over the same ordered task list. The \texttt{experiment/} package holds the batch specification that fixes a configuration before it is launched, and the per-configuration report that computes adherence and its intervals from the run records.

\subsection{Declaring a Collaboration}
\label{app:declare}

A team is declared in one file (Figure~\ref{fig:mechanism}b). The file names the entry role, gives each role a system prompt, the agent profile it starts from, the names of the tools it may call and its token allowance, optionally overrides the role's model, names the context policy every session runs, and lists the directed edges along which one role may message another. Tool names are resolved to implementations from a registry when the team is assembled, so one file yields the same tool set in every run that uses it.

Roles are separated by what they can call, not only by what they are asked to do. A role whose tool set contains no file-editing tool cannot edit through the file tools, whatever its prompt says. A boundary set in the tool set is enforced at call time, so a call outside it is refused rather than left to the model's choice.

The runtime enforces the declaration at the moment of each call. A tool call outside the calling role's tool set, and a message along an edge the file does not declare, are refused, and the refusal is recorded as an event (Appendix~\ref{app:schema}).

The edges decide who can reach whom. In the Star team of Figure~\ref{fig:topologies} the entry role and each Coder can message each other, and the two Coders cannot message each other; in the Ring the analyst can message only the coder, the coder only the tester and the tester only the analyst. Every team needs a path back to the entry role: without one, a teammate could not return its result, and a delegation that did happen but could not be returned would look like a model that never delegated.

The roster is fixed before the first model call and cannot change during a run; a request to add an agent is refused and recorded.

\subsection{Earlier Evaluation Results}
\label{app:earlier-benchmark-results}

The updated benchmark table uses Duo v8, the latest version of Duo included in OpenCollab, and adds OpenHands to the comparison. Table~\ref{tab:benchmark-context-earlier} presents the results previously reported in the paper. In this table, the OC (Duo) entries were obtained using the earlier Duo v6 implementation, which accounts for their differences from the updated results.

\begin{table}[htbp]
    \caption{Earlier cross-harness evaluation results.}
    \label{tab:benchmark-context-earlier}
    \centering
    \footnotesize
    \setlength{\tabcolsep}{5pt}
    \renewcommand{\arraystretch}{1.12}
    \begin{tabular}{@{}llcccc@{}}
        \toprule
        \textbf{Dataset} & \textbf{Harness}
        & \shortstack{\textbf{Pass@1}\\
                      \textbf{(\%)} $\uparrow$}
        & \shortstack{\textbf{Avg. Tokens}\\
                      \textbf{(M)} $\downarrow$}
        & \shortstack{\textbf{Avg. Est. Cost}\\
                      \textbf{(\$)} $\downarrow$}
        & \shortstack{\textbf{Cache Hit}\\
                      \textbf{(\%)} $\uparrow$} \\
        \midrule
        \multirow{5}{*}{SWE-bench Pro}
        & Mini-SWE-Agent & 61.66 & 4.16 & 0.90 & 36.65 \\
        & Codex CLI      & 63.73 & 7.31 & 0.73 & 90.27 \\
        & Claude Code    & 58.03 & 9.38 & 2.54 & 12.50 \\
        & OC (Base)      & 63.21 & 3.89 & 0.41 & 88.03 \\
        & OC (Duo)       & \textbf{64.25}
                         & 6.50 & 0.73 & 84.95 \\
        \midrule
        \multirow{5}{*}{Terminal-Bench 2.1}
        & Mini-SWE-Agent & 76.40 & 2.97 & 0.45 & 67.36 \\
        & Codex CLI      & 80.90 & 2.86 & 0.26 & 93.88 \\
        & Claude Code    & 77.53 & 21.95 & 2.73 & 79.32 \\
        & OC (Base)      & 79.78 & 1.54 & 0.19 & 78.38 \\
        & OC (Duo)       & \textbf{83.15}
                         & 4.60 & 0.60 & 76.47 \\
        \midrule
        \multirow{5}{*}{DeepSWE}
        & Mini-SWE-Agent & 61.95 & 19.36 & 3.27 & 58.81 \\
        & Codex CLI      & 47.79 & 15.01 & 1.30 & 96.47 \\
        & Claude Code    & 56.64 & 77.90 & 7.66 & 91.16 \\
        & OC (Base)      & 55.75 & 11.95 & 1.19 & 90.53 \\
        & OC (Duo)       & \textbf{69.91}
                         & 26.11 & 2.67 & 89.36 \\
        \bottomrule
    \end{tabular}
\end{table}

\subsection{Audit of Ten Agent Artifacts}
\label{app:matrix}

Section~\ref{sec:oc-controlled} names seven conditions for a controlled comparison of agent organizations. Five ask whether a factor can be set explicitly before a run: the model, the tools, the token budget, the context policy and the topology. Two ask whether a run can be checked afterwards against what it was declared to be (\emph{Realized}, \emph{Compared}). Table~\ref{tab:instruments} marks each condition on ten public agent artifacts, with OpenCollab scored by the same rule in the last row. This appendix gives the rule and the reason behind every mark other than \checkmark.

\paragraph{How the marks were placed.} Nine artifacts were read in code at a pinned version; Claude Code is closed-source, and its row rests on vendor documentation. A \checkmark\ records that a researcher sets the factor in a setting the artifact itself reads, such as a configuration file, a declarative specification or a command-line option, and that the artifact applies it to every agent. A \partmark\ records that the artifact ships a mechanism but takes it only in the researcher's own code, or meets the condition only in part. A $\times$ records that no mechanism meets the condition even in part; a setting that only moves the threshold of a built-in behavior or switches it off does not count as one. Each row is read on the mechanism that leaves the organization to the model, such as teams, handoffs, subagents or agents called as tools, as OpenCollab's row is read on its Team controller; a script the researcher fixes before the run is outside the table. Marks follow each artifact's default settings. An artifact that runs a single agent is marked on that agent, and $\times$ under \emph{Topology} and \emph{Compared}. A condition is marked \checkmark\ when:
\begin{itemize}
\item \emph{Model}: every agent's model can be set, and a declared model is not replaced by another without notice;
\item \emph{Tools}: each role can be given its own tool set;
\item \emph{Budget}: each model call is admitted against the calling agent's remaining allowance before it is sent, and refused if it does not fit;
\item \emph{Context}: a setting chooses the policy that decides what each model call sees;
\item \emph{Topology}: the roles and the edges between them are fixed before the run rather than started by the model while it runs;
\item \emph{Realized}: each run's record holds the whole declared organization, including roles that never act, with every model call and tool call attributed to its run and agent;
\item \emph{Compared}: \emph{Realized} holds, every message between agents names its sender and receiver, and every attempt to leave the declaration is refused or flagged in the record.
\end{itemize}
The pinned versions and the file-and-line evidence behind every mark are released with the paper.

\paragraph{Why the other rows fall short.} On \emph{Model} and \emph{Tools}, AG2 and LangGraph take both only as constructor arguments in the researcher's code, HAL passes a model name on to the agent's own code, and a Codex CLI role file can only switch tools off from a fixed set; HAL leaves the tools to the agent. On \emph{Budget}, no external artifact prices a call against the remaining allowance before sending it. AutoGen, SWE-agent, OpenHands, Claude Code and Inspect AI stop an agent once its limit has been reached or after a call has returned; AG2 only raises alerts, LangGraph limits steps, DeepSeek Harness caps the output of a single call, Codex CLI's token budget sits behind a feature marked as not ready for external use, and HAL bounds a task by wall-clock time. On \emph{Context}, AutoGen and SWE-agent choose the policy in configuration; OpenHands, DeepSeek Harness, AG2, LangGraph and Inspect AI let it be chosen only in part or only in code; Claude Code and Codex CLI expose only a threshold and an off switch, and HAL does not run the agent's loop. On \emph{Topology}, AutoGen and Inspect AI hold the agents and whom each may call in a configuration file; AG2 and LangGraph wire agents in code; in Claude Code, Codex CLI, DeepSeek Harness and OpenHands the roles are declared but the model decides at run time which agents to start, and the edges can be held only by a hook the researcher writes; HAL leaves the organization to the agent's own code. On \emph{Realized}, SWE-agent stores its full configuration with every run; seven rows record only part of the declaration, such as the settings of the agents that were started but not of those that were not, or roles without their models; AutoGen's logs carry no run identifier, and LangGraph's traces hold what each call sent but not the graph. On \emph{Compared}, only Inspect AI refuses or flags departures in its record, and only for agents joined as tools; SWE-agent's record is complete, but it runs a single agent.

\paragraph{Our own row.} OpenCollab's row is read on its Team controller. Its team file sets each role's model, tools, token allowance and context policy together with the edges; every call is admitted against the calling role's remaining allowance before it is sent (Appendix~\ref{app:arch}); each run's record opens with the declared organization, every role with its model and tools and every edge, so a declared teammate that never acts still appears; and a message along an undeclared edge or a call outside a role's tool set is refused and recorded (Appendix~\ref{app:schema}).

\subsection{The Run Record}
\label{app:schema}

Section~\ref{sec:oc-observability} replaces prose logs with an ordered event stream. This appendix gives the files each configuration writes, the event types, and the recorded termination reasons.

\paragraph{Three files per configuration.} The batch manifest holds the assigned condition: arm, task list, role cards with their sha256, per-role allowance, the runtime and evaluator versions, and the driver's environment. The event log holds the realized side, one ordered stream per run. The metrics file holds one row per run, with the image the run was generated in. The run id joins the three files, and every event carries it. Rates are computed from individual runs, never from pre-averaged configurations, and adherence is computed from the manifest and the event log after the run (Appendix~\ref{app:axes}).

\paragraph{Events.} Every event carries a timestamp, a step number that increases within the run, the run id, the event type, and a payload whose fields depend on the type and, for every type the analyses below read, name the role. Table~\ref{tab:events} lists the seven types the analyses read and what each one supports; the controllers also write message and budget records. Because the payload names the role, what each role spent, how many steps it took and which tools it called are read directly off the stream, with no prose log to parse.

\begin{table}[t]
\centering
\small
\caption{The seven event types. Model calls and agent finishes settle whether a teammate worked; tool executions of \texttt{message\_agent} show who briefed whom; a session termination, read with the token counts of the model calls, identifies the runs whose decision the allowance cut short.}
\label{tab:events}
\begin{tabular}{@{}>{\raggedright\arraybackslash}p{0.19\textwidth}>{\raggedright\arraybackslash}p{0.33\textwidth}>{\raggedright\arraybackslash}p{0.40\textwidth}@{}}
\toprule
\textbf{Event type} & \textbf{What it adds to the common fields} &
\textbf{What it supports} \\
\midrule
Model call & Token counts and latency of the step & Per-role spend; the
organizational overhead of Appendix~\ref{app:cost-metrics} \\
\addlinespace
Tool execution & Tool name, arguments and result & Who briefed whom; which
edges were addressed. Refused messages and refused tool calls are written as
separate events \\
\addlinespace
Agent finish & Whether the role produced content & Whether a teammate
worked, read together with its model calls \\
\addlinespace
Session termination & The role's termination reason
(Table~\ref{tab:termination}) & Which runs were cut off by the allowance
after delegation began \\
\addlinespace
Refused spawn & The reason for the refusal & That the roster cannot change
within a run (Section~\ref{sec:oc-declare}); that a single-agent run cannot
coordinate \\
\addlinespace
Worktree change & Commit sha & Which role's commits reached which worktree
(Section~\ref{sec:oc-observability}) \\
\addlinespace
Context shaping & Which compaction stage fired this turn & Whether a call's
input was compacted (Section~\ref{sec:oc-controlled}) \\
\bottomrule
\end{tabular}
\end{table}

\begin{table}[t]
\centering
\footnotesize
\caption{The end states of a session (Figure~\ref{fig:lifecycle}) and the reasons recorded with them. Each reason is a string whose fixed prefix is shown; some continue with a count, such as the tokens used. The wall-clock limit is not among them: the evaluator stops a run that reaches it from outside the runtime and records that in the run's metrics row.}
\label{tab:termination}
\begin{tabular}{@{}l>{\raggedright\arraybackslash}p{0.36\textwidth}>{\raggedright\arraybackslash}p{0.46\textwidth}@{}}
\toprule
\textbf{State} & \textbf{Recorded reason} & \textbf{When it is recorded} \\
\midrule
\texttt{DONE} & \texttt{completed}; \texttt{submitted} & The role gave a final
answer, or called the submit tool. \\
\midrule
\texttt{STOPPED} & \texttt{budget exceeded} & At \texttt{PRECHECK}, the role
had already spent its allowance. \\
 & \texttt{budget exhausted before model call} & The next call, priced before
it was sent, would not fit in what remained. \\
 & \texttt{budget exceeded after model call} & The last call cost more than its
estimate and took the role past its allowance. \\
 & \texttt{step limit reached} & The session reached its step limit. \\
 & \texttt{loop block limit reached} & The role kept repeating a tool call
without progress. \\
 & \texttt{context overflow} & The prompt exceeded the model's context window
even after compaction. \\
 & \texttt{output truncated} & The provider cut the response off at its
generation limit. \\
 & \texttt{required tool was not called after correction} & The role was
required to call a tool, was reminded once, and still did not. \\
 & \texttt{interrupted by user}; \texttt{cancelled} & The session was stopped
from outside. \\
\midrule
\texttt{ERROR} & the exception's type and message & An unhandled fault,
including a role whose container or workspace was withdrawn during the run. \\
\bottomrule
\end{tabular}
\end{table}

\paragraph{Cost metrics.}\label{app:cost-metrics}

\textbf{Organizational overhead.} Multi-agent configurations incur prompt costs even when the assigned organization is not fully realized. Because Table~\ref{tab:events} logs token usage by role, OpenCollab measures this prefix tax by comparing, on the same tasks, the cost of Team runs with $A_r = 0$ against the cost of the single-agent baseline. With $C(z)$ and $C(0)$ the potential total costs under target Team configuration $z$ and Single, defined like $Y(z)$ and $Y(0)$ in Appendix~\ref{mas:setup},
\begin{equation}
\begin{gathered}
	\mathrm{Overhead}_0 := \mathbb{E}[C(z) - C(0) \mid A(z) = 0], \\
	\widehat{\mathrm{Overhead}}_0 = \frac{1}{n_0} \sum_{i:\,A_i(z) = 0} \bigl(C_i(z) - C_i(0)\bigr),
\end{gathered}
\end{equation}
where the sum runs over the $n_0$ tasks whose Team run has $A_r = 0$ and pairs each with the Single run of the same task. This metric is the additional cost of carrying the Team configuration among runs that do not fully realize it. Comparing these Team runs with the unconditional Single mean would instead mix that cost with differences in task composition, since non-adherent runs need not fall on typical tasks.

\textbf{Workload balance.} For fully adherent Team runs ($A_r = 1$), OpenCollab measures how computational effort is divided across the active roles. Letting $p_s(r) = C_s(r) / C(r)$ denote the token proportion consumed by role $s$, the framework defines workload balance using normalized Shannon entropy across the active roles $\mathcal{S}$:
\begin{equation}
	\mathrm{Balance}(r) = -\sum_{s \in \mathcal{S}} \frac{p_s(r) \log p_s(r)}{\log |\mathcal{S}|}.
\end{equation}
The aggregate balance metric is computed as the expectation of $\mathrm{Balance}(r)$ across all adherent runs where $A_r = 1$. A value near one reflects an egalitarian division of labor, whereas a value near zero indicates that a single role dominated execution despite the presence of teammates.

\subsection{Adherence, Axis by Axis}
\label{app:axes}

Section~\ref{sec:oc-adherence} defines a run as adherent when every axis that applies to its arm is adherent. This appendix states the six axes (Table~\ref{tab:adherence}), how a verdict is reached, and which axes can vary in these experiments.

\paragraph{How a verdict is reached.} Each run has an assigned side, written in the configuration's manifest before the run and restated as events in the run's event log, and a realized side, the rest of that log. Each axis compares the two and returns \emph{adherent}, \emph{deviant} or \emph{unverifiable}. $\widehat{\alpha}_{\mathrm{adh}}$ is the share of a configuration's runs that are adherent on every axis that applies. An axis that does not apply to an arm, such as delegation for the single agent, is left out of that arm's conjunction rather than scored as a pass.

\paragraph{Which axes can vary here.} An adherent verdict says that a run did not depart from its declaration on that axis, which is not the same as the axis having had anything to measure. Four of the six axes are held by the runtime in these experiments. The runtime offers each agent only the tools of its assigned set and refuses any other call, so no agent can act outside it, which is the one deviation the role-boundary axis records. An allowance refuses the call that would exceed it and records the refusal, so no role can pass its allowance unrecorded, the budget-sharing deviation. A message along an edge the topology does not declare is refused. Every role runs the same context policy. In the Team the conjunction is therefore decided by participation and delegation. A teammate starts working only when a message reaches it (Appendix~\ref{app:arch}), so a run in which every declared teammate worked has also delegated, and the conjunction comes down to one test: a Team run is adherent when every declared teammate, every role other than the entry role, spent tokens and produced at least one model output. This is computed for every run from the run record: the roles from the declaration restated in the event log, and each role's tokens and outputs from its saved session. The four held axes still have to be declared: a comparison that gives each model a different tool set, allowance or context window is exactly the case they exist for.

\paragraph{Unverifiable counts against adherence.} When the log cannot establish what happened on an axis, the verdict is \emph{unverifiable} and the run is not counted as adherent. This happens, for instance, when a run was cut off after the entry agent had begun to delegate but before every teammate had worked (Appendix~\ref{app:protocol} lists the cases). Such a run has $A_r = 0$ (Appendix~\ref{mas:setup}). Missing evidence can therefore lower a reported rate but never raise it; counting every unverifiable run as adherent instead shows how much the rate depends on this rule.

\begin{table}[t]
\caption{The six adherence axes. The assigned value is written before the run, and the realized value is read from the event log.}
\label{tab:adherence}
\centering
\small
\begin{tabularx}{\textwidth}{lYY}
\toprule
\textbf{Axis} & \textbf{Realized value, read from} & \textbf{Deviation} \\
\midrule
Participation & the teammates that spent tokens and produced at least one
model output & a declared teammate did no work \\
Delegation & a teammate that spent tokens and produced at least one model
output & no such teammate, where the assigned regime offers or requires
delegation \\
Role boundary & the tool calls attributed to each agent & a call outside the
agent's assigned tool set is executed \\
Budget sharing & each role's token draw against its own allowance & a role
reached its allowance and no refusal was recorded \\
Information flow & \texttt{message\_agent} calls in the event log, in every
arm & a message travels along an edge the topology does not declare \\
Context policy & the compaction level applied on each turn & a level other
than the assigned one fires \\
\bottomrule
\end{tabularx}
\end{table}

\end{document}